\documentclass[10pt, a4paper,reqno]{amsart}
\usepackage{latexsym, amsfonts, euscript, amsmath, amssymb, amsfonts}
\newtheorem{theorem}{Theorem}
\newtheorem{proposition}[theorem]{Proposition}
\newtheorem{lemma}[theorem]{Lemma}
\newtheorem{corollary}[theorem]{Corollary}

\newtheorem*{theorem*}{Theorem}
\newtheorem*{conjecture*}{Conjecture}

\theoremstyle{remark}
\newtheorem{remark}[theorem]{Remark}

\newtheorem{example}[theorem]{Example}

\newcommand{\GL}{{\mathrm{GL}}}
\newcommand{\AGL}{{\mathrm{AGL}}}

\newcommand{\w}{{\mathrm{w_H}}}

\newcommand{\sgn}{{\mathrm{sgn}}}
\newcommand{\Aut}{{\mathrm{Aut}}}

\newcommand{\Aff}{\mathbb{A}}
\newcommand{\Lr}{\mathcal{L}_r}
\newcommand{\Sr}{\mathrm{S}_r}
\newcommand{\Si}{\mathrm{S}_i}
\newcommand{\Sj}{\mathrm{S}_j}
\newcommand{\M}{\mathcal{M}}
\newcommand{\basis}{\EuScript{B}(\ell,m;r)}

\newcommand{\fstar}{\EuScript{F}^*(\ell,m;r)}

\newcommand{\fstarmin}{\EuScript{F}^*_{\rm min}(\ell,m;r)}

\newcommand{\tsM}{{\mathsf{t}}_{\sigma}({\mathcal{M}})}%
\newcommand{\tpM}{{\mathsf{t}}_{\pi}({\mathcal{N}})}

\newcommand{\BtM}{{\mathsf{B}}_{{\mathcal{M}}, \sigma}}
\newcommand{\teM}{{\mathsf{t}}_{\epsilon}({\mathcal{M})}}
\newcommand{\ttM}{{\mathsf{t}}_{\tau}({\mathcal{M})}}
\newcommand{\Ll}{\mathcal{L}_{\ell}}
\newcommand{\LM}{\mathcal{L}}

\newcommand{\Z}{\mathbb{Z}}

\newcommand{\Adelta}{\mathbb{A}^{\delta}\left({\mathbb F}_q\right)}
\newcommand{\Ad}{\mathbb{A}^{\delta}}
\newcommand{\F}{\mathbb{F}_q}
\newcommand{\Fq}{\mathbb{F}_q}
\newcommand{\fqd}{{\mathbb{F}_q}[T]_{\le d}}
\newcommand{\fqds}{{\mathbb{F}_q}[T_1, \dots , T_s]_{\le (d_1, \dots , d_s)}}
\newcommand{\mdts}{M_{(d_1, \dots , d_s)}[T_1, \dots , T_s]}
\newcommand{\Fqm}{\mathbb{F}_q^m}
\newcommand{\C}{C^\mathbb{A}(\ell,m)}
\newcommand{\Cr}{C^\mathbb{A}(\ell,m;r)}
\newcommand{\Cl}{C^\mathbb{A}(\ell,m;\ell)}
\newcommand{\Mon}{\mathbb{M}(\ell,m)}
\newcommand{\RMon}{\overline{\mathbb{M}}(\ell,m)}
\newcommand{\FM}{\overline{\mathrm{F}\mathbb{M}}(\ell,m;r)}
\newcommand{\rp}{\mathfrak{R}(\ell,m)}
\newcommand{\Ev}{\mathrm{Ev}}
\newcommand{\RM}{\mathrm{RM}}
\newcommand{\Term}{\mathrm{Term}}

\newcommand{\supp}{\mathrm{supp}}
\newcommand{\full}{\mathsf{F}}
\newcommand{\Bnot}{\mathsf{B}_0}
\newcommand{\rec}{\mathsf{R}}
\begin{document}

\title{Duals of Affine Grassmann Codes and their Relatives}

\author{Peter Beelen}
\address{Department of Mathematics, Technical University of Denmark, \newline \indent
DK 2800, Lyngby, Denmark.}
\email{p.beelen@mat.dtu.dk}

\author{Sudhir R. Ghorpade}
\address{Department of Mathematics,
Indian Institute of Technology Bombay,\newline \indent
Powai, Mumbai 400076, India.}
\email{srg@math.iitb.ac.in}

\author{Tom H{\o}holdt}
\address{Department of Mathematics, Technical University of Denmark, \newline \indent
DK 2800, Lyngby, Denmark.}
\email{T.Hoeholdt@mat.dtu.dk}

\date{\today}

\begin{abstract}
Affine Grassmann codes are a variant of generalized
Reed-Muller codes and are closely related to Grassmann codes. These codes were introduced in a recent work
\cite{BGH}. Here we consider, more generally, affine Grassmann codes of a given level. We explicitly determine the
dual of an affine Grassmann code of any level and compute its minimum distance.
Further, we ameliorate the results of \cite{BGH} concerning the automorphism group of affine Grassmann codes.
Finally, we prove that affine Grassmann codes 
and their duals have the property that they are linear codes generated by their minimum-weight codewords.
This provides a clean analogue of a corresponding result for generalized Reed-Muller codes.

\end{abstract}
\maketitle

\section{Introduction}

Fix 
a finite field $\F$ with $q$ elements
and positive integers $\ell, \ell'$ with $\ell\le \ell'$; set
$$
m = \ell + \ell' \quad \mbox{and} \quad \delta = \ell \ell'.
$$
Briefly put, the affine Grassmann code $\C$ is the $q$-ary linear code obtained by evaluating
linear polynomials in the minors of a generic
$\ell \times \ell'$ matrix $X$ at all points of the $\delta$-dimensional affine space of $\ell \times \ell'$ matrices with entries in $\F$. Evidently, when $\ell =1$, this gives the first
order generalized Reed-Muller code ${\rm RM}(1,\ell')$.  However, in general, $\C$ is only a subcode of the $\ell^{\rm th}$ order generalized Reed-Muller code ${\rm RM}(\ell, \delta)$. The length $n$ and the dimension $k$ of $\C$ are given by
$$
n = q^{\delta} \quad \mbox{and} \quad k = \binom{m}{\ell}.
$$
Affine Grassmann codes were introduced in \cite{BGH}, where the following 
was shown.
\begin{itemize}
	\item The minimum distance of $\C$ is
	\begin{equation}
\label{eq:dlm}
	d(\ell,m) := q^{\delta-\ell^2}\prod_{j=0}^{\ell-1}(q^\ell-q^j) = q^{\delta}\prod_{i=1}^{\ell}\left(1 - \frac{1}{q^i}\right).
	\end{equation}
\item
The (permutation) automorphism group of 
$\C$ contains a subgroup isomorphic to the semidirect product $M_{\ell\times\ell'}(\F) \rtimes_{\theta} \GL_{\ell'}(\F)$ of the additive group of $\ell\times\ell'$ matrices over $\F$ with the multiplication group of $\ell'\times\ell'$ nonsingular matrices over $\F$, where  $\theta: GL_{\ell'}(\F) \to \Aut(M_{\ell\times\ell'}(\F))$ is the homomorphism defined by $\theta(A)(\mathbf{u}):= \mathbf{u}A^{-1}$.
\item
The minimum-weight codewords of $\C$ are precisely the evaluations of leading maximal minors (formed by the $\ell$ rows and the first $\ell$ columns) of
$X'$, where $X'=XA^{-1}+\mathbf{u}$ for some $A \in \GL_{\ell'}(\F)$ and $\mathbf{u} \in M_{\ell \times \ell'}(\F)$.
\item
The number of minimum-weight codewords of $\C$ is given, in terms of the Gaussian binomial coefficients (defined below for any $a\ge b\ge0$),
by
\begin{equation}
\label{eq:Ad}
(q-1)q^{\ell^2} {{\ell'}\brack{\ell}}_q \quad \text{where} \quad
{{a}\brack{b}}_q := \frac{(q^{a}-1)(q^{a}-q)\cdots (q^{a}-q^{{b}-1})}{(q^{b}-1)(q^{b}-q)\cdots (q^{b}-q^{{b}-1})}.
\end{equation}
\end{itemize}
In this paper we continue the study of affine Grassmann codes and give an explicit description of the dual of $\C$. As a result, it will be seen that affine Grassmann codes are almost always self-orthogonal. Moreover, we
determine precisely the minimum distance
of ${\C}^{\perp}$ and show that it is at most $4$.
Thus, it is seen that the parity check matrix of $\C$ is rather sparse and that an affine Grassmann code
may be regarded as a LDPC code. Further, following a suggestion by an anonymous referee of \cite{BGH},
we augment the abovementioned result on the automorphism group of $\C$
by showing that $\Aut(C)$ contains, in fact, a larger group that is essentially obtained by taking the product
of the general linear group $\GL_{\ell}(\F)$ with
the semidirect product $M_{\ell\times\ell'}(\F) \rtimes_{\theta} \GL_{\ell'}(\F)$. It will also be seen that
the full automorphism group can, in fact, be even larger. Finally, we show that the affine Grassmann codes as well as their duals have the property that the minimum-weight codewords generate the code. This can be viewed as an analogue of the classical result that binary Reed-Muller codes are generated by their minimum-weight codewords; see, e.g., MacWilliams and Sloane \cite[Ch. 13, \S 6]{MW}. Such a result is not true, in general, for $q$-ary generalized Reed-Muller codes, and in this case, a complete characterization of generation by the  minimum-weight codewords was obtained by Ding and Key \cite[Thm. 1]{DK} (see also part (v) of Proposition~\ref{pro:RMbasics} below). 
A special case $\ell=1$ of our results corresponds to their result 
for the generalized Reed-Muller codes $\RM(1,\, \delta)$ and $\RM(\delta(q-1)-2, \, \delta)$.

Following a suggestion of D. Augot, we shall consider in this paper a mild generalization of $\C$ obtained by choosing a nonnegative integer $r\le \ell$ and then restricting the function space to linear polynomials in the $i\times i$ minors of $X$ for $i\le r$. 
The resulting linear codes are denoted by $\Cr$
and called \emph{affine Grassmann codes} of \emph{level $r$}. Note that the first order Reed-Muller codes of length $q^{\delta}$ as well as the affine Grassmann codes are special cases; indeed,  $C^\mathbb{A}(\ell,m;1) = {\rm RM}(1, \delta)$ and $\Cl = \C$.
Moreover, by varying the levels, we obtain a nice filtration,  compatible with the  Reed-Muller filtration:

\smallskip

\begin{center}
\begin{tabular}{ccccccc}
$C^\mathbb{A}(\ell,m;1)$ & \! \! $\subset$ \! \!  & $C^\mathbb{A}(\ell,m;2)$ &\! \! $\subset$ \! \!& \! \! $\cdots$ \!\! & \!\! $\subset$ \!\!&  $C^\mathbb{A}(\ell,m;\ell)$ \\
$\|$ & & $\bigcap$ & &  & & $\bigcap$ \\
$\RM(1, \delta)$ & \! \! $\subset$ \! \! & $\RM(2, \delta)$ & \! \! $\subset$ \! \! & \! \! $\cdots$ & \! \! $\subset$ \! \! & $\RM(\ell, \delta)$
 \end{tabular}
\end{center}

In general, for any
$r\ge 0$, the length $n$ and the dimension $k_r$ of $\Cr$ are given by
\begin{equation}
\label{eq:nkr}
n = q^{\delta} \quad \mbox{and} \quad k_r = \sum_{i=0}^{r} {\binom{\ell}{i}}{\binom{\ell'}{i}},
	\end{equation}
whereas the formula \eqref{eq:dlm} generalizes nicely to the following:
\begin{equation}
\label{eq:dlmr}
\text{minimum distance of } \Cr  = q^{\delta}\prod_{i=1}^{r}\left(1 - \frac{1}{q^i}\right).
	\end{equation}
The augmentation of the result concerning the automorphism group, an explicit description of the dual, determination of the minimum distance of the dual, and the result concerning generation by minimum-weight codewords are all obtained more generally, in the case of affine Grassmann codes of any given level. However, for the duals $\Cr^\perp$, it is shown that generation by minimum-weight codewords is valid for $r=1$ and $r=\ell$, but not, in general, for $1<r<\ell$.

\section{Preliminaries}
\label{sec:prelim}

Let $X= \left(X_{ij}\right)$ be  a $\ell \times \ell'$ matrix whose entries are algebraically independent indeterminates over $\F$. By $\rec$ we denote the integral rectangle $[1,\ell]\times [1, \ell']$, i.e.,
$$
\rec :=\left\{(i,j)\in \Z^2 : 1\le i\le \ell \text{ and }  1\le j \le {\ell}'\right\}.
$$
By $\F[X]$ we denote the polynomial ring in the  $\ell\ell'$ variables $X_{ij}$  (where $(i,j)$ vary over $\rec$)
with coefficients in $\F$. The set of all monomials in $\F[X]$ will be denoted by $\Mon$. Note that every $\mu \in \Mon$ is of the form
$$
\mu = 
\prod_{(i,j)\in \rec} X_{ij}^{\alpha_{ij}}
\quad \text{for some nonnegative integers }\alpha_{ij} 
$$
The exponents $\alpha_{ij}$ ($(i,j)\in \rec$) are uniquely determined by $\mu$ and their sum is denoted by $\deg \mu$; also, we write $\deg_{X_{ij}}\mu = \alpha_{ij}$.
We say that the monomial $\mu$ is  \emph{reduced} (resp: \emph{squarefree}) if $0\le   \deg_{X_{ij}}\mu  \le q-1$ 
(resp: $0\le   \deg_{X_{ij}}\mu  \le 1$)
for all $(i,j)\in \rec$. 
These two notions coincide when $q=2$. The set of
all reduced monomials in $\F[X]$ will be denoted by $\RMon$ and the $\F$-linear space generated by $\RMon$ will be denoted by $\rp$. Elements of $\rp$ are called \emph{reduced polynomials}. There is a natural surjective map from $\Mon$ to $\RMon$ that sends a monomial $\mu$
to the unique monomial $\bar\mu$ obtained from $\mu$ as follows: whenever an exponent $ \alpha_{ij}$ of $X_{ij}$ is $\ge q$,  replace it by $r_{ij}$, where $r_{ij}\equiv  \alpha_{ij} ({\rm mod} \; q-1)$ and $1\le r_{ij} \le q-1$. 
This map extends by $\F$-linearity to a surjective $\F$-vector space homomorphism $\F[X]\to \rp$, which may be referred to as the \emph{reduction  map}. We will denote the image of 
$f\in \F[X]$ under the reduction map by $\bar f$, and call $\bar f$ the \emph{reduced polynomial} corresponding to $f$.

The set $\Mon$ is obviously a $\F$-basis of $\F[X]$ and hence every $f\in \F[X]$ can be uniquely written as $\sum_{\mu\in\Mon} c_{\mu}\mu$, where $c_{\mu}\in \F$ for each $\mu\in \Mon$ and $c_{\mu}=0$ for all except finitely many $\mu$'s. A monomial $\mu$ for which $c_{\mu}\ne 0$ will be referred to as a \emph{term} of $f$, and we let
$$
\Term(f):=\left\{\mu\in \Mon : c_{\mu}\ne 0\right\}.
$$
Note that $\Term(f)$ is the empty set if and only if $f$ is the zero polynomial. For $0\ne f\in \F[X]$, the (total) degree and the degree in the 
variable $X_{ij}$ are given by
$$
\deg f  := \max\{\deg \mu : \mu \in \Term(f)\}  \mbox{ and }   \deg_{X_{ij}} f := \max\{\deg_{X_{ij}}\mu : \mu \in \Term(f)\}.
$$

We shall denote the space of all $\ell\times \ell'$ matrices with entries in $\F$ by $\Adelta$, or simply by ${\mathbb A}^{\delta}$.
Fix an enumeration $P_1,P_2,\dots,P_{q^\delta}$ of $\Aff^{\delta}$. The map
$$
\Ev: \F[X]\to  \F^{q^{\delta}} \quad \mbox{ defined by } \quad
\Ev(f):= \left(f(P_1),\dots,f(P_{q^\delta}) \right)
$$
will be referred to as the \emph{evaluation map} of $\F[X]$.
%
It is clear that the evaluation map $\Ev$ defined above is a surjective linear map, and also that $\Ev(f) =\Ev(\bar f)$ for every $f\in \F[X]$.
Thus, the restriction of $\Ev$ to $\rp$ is also surjective. In fact,  it is well-known that this restriction is injective as well.
(See, e.g., \cite[p. 11]{Joly}.)
In other words, reduced polynomials can be identified with functions from $\Aff^{\delta}$ to $\Fq$.

\begin{remark}
\label{multred}
Although the reduction map from $\Fq[X]$ onto $\rp$  
is $\Fq$-linear, it is not multiplicative, i.e., $\overline{fg}$ need not be equal to $\bar{f}\bar{g}$, in general.
In fact, the product of reduced monomials need not be a reduced monomial.
However, if $f, g\in \Fq[X]$ are polynomials in disjoint sets of variables, then $\overline{fg} = \bar{f}\bar{g}$.
\end{remark}

Recall that by a \emph{minor} of $X$ of order $i$ we mean the determinant of an $i\times i$ submatrix of $X$. A minor of $X$ of order $i$ is
sometimes referred to as an $i \times i$ minor of $X$.
For $0\le i\le \ell$, let $\Delta_i(\ell,m)$ be the subset of $\F[X]$ consisting of all $i \times i$ minors of $X$, where, as per standard conventions, the only $0\times 0$ minor of $X$ is $1$.  For $0\le r\le \ell$, we define
$$
\Delta(\ell,m;r) := \bigcup_{i=0}^{r} \Delta_i(\ell,m) 
$$
and ${\mathcal F}(\ell,m;r)$ to be the $\F$-linear  subspace of $\F[X]$ generated by $\Delta(\ell,m;r)$.
Often $\Delta(\ell,m; \ell)$ and ${\mathcal F}(\ell,m;\ell)$ will just be denoted by $\Delta(\ell,m)$ and ${\mathcal F}(\ell,m)$, respectively.
Observe that $\deg_{X_{ij}} {\mathcal M} \le 1 $ for all
${\mathcal M}\in \Delta(\ell,m)$ and  $(i,j)\in \rec$.
In particular,  ${\mathcal F}(\ell,m) \subseteq \rp$. Next, we record the 
following basic result. It is  an
easy consequence of Lemma 2 of \cite{BGH} and its proof together with Lemma 3 of \cite{BGH}.

\begin{proposition}
\label{pro:dimFmlr}
For every $r\in \{0,1, \dots, \ell\}$, the elements of $\Delta(\ell,m;r)$ are linearly independent.
In particular,
$$
\dim_{\F}{\mathcal F}(\ell,m;r)= \sum_{i=0}^{r} {\binom{\ell}{i}}{\binom{\ell'}{i}}
\quad \mbox{and} \quad
\dim_{\F}{\mathcal F}(\ell,m)=\binom{m}{\ell}.
$$
\end{proposition}

Thanks to Proposition \ref{pro:dimFmlr}, every  $f\in {\mathcal F}(\ell,m)$ is a unique $\F$-linear combination of the elements of $\Delta(\ell,m)$,
say  $f=\sum_{{\mathcal M} \in \Delta(\ell,m)}a_{\mathcal M}{\mathcal M}$, where $a_{\mathcal M}\in \F$ for every ${\mathcal M}\in{\Delta}(\ell,m)$. We define
the \emph{support} of $f$ 
to be the set 
$$\supp(f):=\{{\mathcal M}\in \Delta(\ell,m) \, : \, a_{{\mathcal M}} \neq 0\}.$$
Note that the support of $f$ is the empty set if and only if $f$ is the zero polynomial. Also note that for $0\le r\le \ell$ and $f\in {\mathcal F}(\ell,m;r)$, the sets $\supp(f)$ and $\Term(f)$ coincide only when $r\le 1$.

For any  
nonnegative integer $r\le \ell$,
the image of ${\mathcal F}(\ell,m;r)$ under the evaluation map $\Ev$ will be denoted by
$\Cr$ and called the \emph{affine Grassmann code of level} $r$. As in \cite{BGH}, we will write $\C=\Cl$ and refer to this simply as the
\emph{affine Grassmann code} (corresponding to the fixed parameters $\ell$ and $\ell'$, or equivalently, $\ell$ and $m$).
The following result is a consequence of Proposition \ref{pro:dimFmlr}. 
Its proof is similar to that of Lemma 7 of \cite{BGH}, 
and is hence omitted.

\begin{proposition}
For each $r\in \{0,1, \dots, \ell\}$,
the affine Grassmann code of level $r$ 
is a nondegenerate linear code of length $n$ 
and dimension $k_r$ given by \eqref{eq:nkr}.
\end{proposition}
%

Finally, in this section we review some basic facts about generalized Reed-Muller codes, which will be useful in the sequel.
First, recall that for any nonnegative integer $r \le \delta(q-1)$, the $r^{\rm th}$ order generalized Reed-Muller code of length $q^{\delta}$,
denoted $\RM_q(r, \delta)$ or simply $\RM(r, \delta)$, is the image of $\{\mu \in \rp: \deg \mu \le r\}$ under the evaluation map $\Ev$.
Some of its fundamental properties are the following.

\begin{proposition}
\label{pro:RMbasics}
Let $r$ be a  nonnegative integer $\le \delta(q-1)$, and let $Q, R$ be unique integers such that $\delta(q-1)-r = Q(q-1)+R$ and $0\le R < q-1$. Then:  
\begin{enumerate}
	\item[{\rm (i)}] 
$\RM(r, \delta)$ is nondegenerate linear code of length $q^{\delta}$ and
$$
\dim \RM(r, \delta) = \sum_{i=0}^r \sum_{j=0}^{\delta} (-1)^j {\binom{\delta}{j}}{\binom{\delta+i-jq-1}{i-jq}}.
$$
In particular, if $r\le q-1$, then the dimension of $\RM(r, \delta)$ is ${\binom{\delta+r}{r}}$.
\smallskip
\item[{\rm (ii)}] The minimum distance of $\, \RM(r, \delta)$ is $(R+1)q^Q$ and the number of minimum-weight codewords of $\, \RM(r, \delta)$ is
given, in terms of the Gaussian binomial coefficients (defined in \eqref{eq:Ad} above), by
$$
\begin{cases}
\displaystyle{\left(q^{\delta-Q+1}-q^{\delta-Q}\right){{\delta}\brack{Q}}_q} & \text{ if } \; R=0, \\ \\
\displaystyle{\left(q^{\delta}-q^{\delta-Q-1}\right){{\delta}\brack{Q+1}}_q \binom{q}{R+1}} & \text{ if } \; R>0.
\end{cases}
$$
\smallskip
\item[{\rm (iii)}]
If $r\ge 1$, then the (permutation) automorphism group of $\, \RM(r, \delta)$ is isomorphic to the affine general linear 
group $\AGL_{\delta}(\Fq)$ of 
transformations $\F^{\delta}\to \F^{\delta}$ of the form $\mathbf{x}\mapsto M\mathbf{x}+\mathbf{u}$, where $M\in \GL_{\delta}(\Fq)$ and $\mathbf{u}\in \F^{\delta}$.
\smallskip
\item[{\rm (iv)}]
The dual of $\, \RM(r, \delta)$ is 
$\, \RM(\delta(q-1)-r-1, \delta)$.
\smallskip
\item[{\rm (v)}]
Write $q=p^t$, where $p$ is a prime number and $t\ge 1$. Then
$\RM(r, \delta)$ is generated by its minimum-weight codewords if and only if 
$\, \delta=1$ or $t=1$ or $r < p$ or $r>(\delta -1)(q-1) + p^{t-1} -2$.
\end{enumerate}
\end{proposition}

A proof of the assertions in 
Proposition \ref{pro:RMbasics} can be found, for example, in: \cite[\S 5.4]{AK} (parts (i) and (iv)),
\cite{DGM} (part (ii)),  \cite{BC} (part (iii)), and \cite{DK} (part (v)).

\section{Minimum distance}
\label{sec:mindist}

For a positive integer $r\le \ell$, we shall denote by $\Lr$ the $r^{\rm th}$ leading principal minor of the $\ell\times \ell'$ matrix $X$.
In other words, $\Lr$ is the determinant of the submatrix of $X$ formed by the first $r$ rows and the first $r$ columns. 
Also, we set $\LM_0 := 1$. 
Often we write $\Ll$ simply as $\LM$ and refer to it as the \emph{leading maximal minor}.

\begin{theorem}
\label{thm:mindistClmr}
Let $r$ be a nonnegative integer $\le \ell$. Then \begin{equation}
\label{eq:dlmr2}
\text{minimum distance of } \Cr  = q^{\delta}\prod_{i=1}^{r}\left(1 - \frac{1}{q^i}\right).
\end{equation}
Moreover, $\Ev\left(\Lr\right)$ is a minimum-weight codeword of $\Cr$.
\end{theorem}

\begin{proof}
Let $f\in {\mathcal F}(\ell,m;r)$ be such that $f\ne 0$. Then there is a nonnegative integer $s\le r$ such that $\supp(f)\cap \Delta_s(\ell,m)$ is nonempty, but $\supp(f)\cap \Delta_i(\ell,m)$ is empty for each $i>s$. 
Choose a minor $\M\in \supp(f)\cap \Delta_s(\ell,m)$ and let $Y$ be the corresponding $s\times s$ submatrix of $X$. In view of Proposition \ref{pro:dimFmlr}, any specialization $\tilde f$ of $f$, obtained by substituting arbitrary values in $\Fq$ for the $\delta - s^2$ variables not occurring in $Y$, is a nonzero linear combination of minors of $Y$. It follows that 
$$
\w\left(\Ev(f)\right) \ge d(s, 2s) q^{\delta - s^2},
$$
where $\w(c)$ denotes the the (Hamming) weight of a codeword $c$ and $d(s,2s)$ denotes the minimum distance of the affine Grassmann code $C^\mathbb{A}(s,2s)$ corresponding to the $s\times s$ matrix $Y$. Using \eqref{eq:dlm} (i.e., Theorem 16 of \cite{BGH}) with $\ell=\ell'=s$, we see that
$$
\w\left(\Ev(f)\right) \ge \left(q^{s^2}\prod_{i=1}^{s}\left(1 - \frac{1}{q^i}\right)\right) q^{\delta - s^2}
\ge q^{\delta}\prod_{i=1}^{r}\left(1 - \frac{1}{q^i}\right).
$$
On the other hand, it is readily seen that $\Lr \in {\mathcal F}(\ell,m;r)$ and
$$
\w\left(\Ev\left(\Lr\right)\right) = q^{\delta - r^2} \# \GL_r(\F) =  q^{\delta - r^2}\prod_{j=0}^{r-1}\left(q^r-q^j\right) = q^{\delta}\prod_{i=1}^{r}\left(1 - \frac{1}{q^i}\right).
$$
This yields \eqref{eq:dlmr2} and also shows that $\Ev\left(\Lr\right)$ is a minimum-weight codeword. 
\end{proof}

It may be tempting to believe that, as in the case of affine Grassmann codes, every minimum-weight codeword  of $\Cr$ is essentially of the form $\Lr$, i.e., it is equal to $\Ev\left({\mathcal{L}'}\right)$, where
${\mathcal{L}'}$ is the $r^{\rm th}$ leading principal minor
of the $\ell\times \ell'$ matrix $X'$, where $X'= XA^{-1} + \mathbf{u}$ for some $A\in \GL_{\ell'}(\F)$ and $\mathbf{u} \in M_{\ell \times \ell'}(\F)$. However, the following example shows that if $r< \ell$, then this need not be the case even when $X'$ is, more generally, of the form
$BXA^{-1} + \mathbf{u}$, where $A, \mathbf{u}$ are as above and $B\in \GL_{\ell}(\F)$.

\begin{example}
Assume that $\ell \ge 2$ and let $\left(c_{ij}\right)$ be any $\ell\times \ell'$ matrix over $\Fq$ of rank $\ge 2$. Then some $2\times 2$ minor of $\left(c_{ij}\right)$ is nonzero. Consider $C^\mathbb{A}(\ell,m;1) = {\rm RM}(1, \delta)$. We know from Reed-Muller theory (or alternatively, Remark 11 of \cite{BGH}) that any linear polynomial in $\F[X]$ in which some $X_{ij}$ occurs with a nonzero coefficient gives rise to a minimum-weight codeword. In particular, $\Ev(f)$ is a minimum-weight codeword of $C^\mathbb{A}(\ell,m;1)$, where
$$
f = \sum_{i=1}^{\ell} \sum_{j=1}^{\ell'} c_{ij}X_{ij}.
$$
However, $f$ is not the first leading principal minor of any $\ell\times \ell'$ matrix of the form $BXA + \mathbf{u}$. Indeed if this were the case for some $B=\left(b_{ij}\right)\in \GL_{\ell}(\F)$, $A=\left(a_{ij}\right)\in \GL_{\ell'}(\F)$, and $\mathbf{u}=\left(u_{ij}\right)\in M_{\ell \times \ell'}(\F)$, then
$$
 \sum_{i=1}^{\ell} \sum_{j=1}^{\ell'} c_{ij}X_{ij} = u_{11} +  \sum_{i=1}^{\ell} \sum_{j=1}^{\ell'} b_{1i}X_{ij}a_{j1}.
$$
Consequently, $u_{11}=0$ and $ c_{ij} =b_{1i}a_{j1}$ for $1\le i\le \ell$ and $1\le j \le\ell'$. But this is a contradiction since it is readily seen that every $2\times 2$ minor of the $\ell\times \ell'$ matrix $\left( b_{1i}a_{j1}\right)$ is always zero.
\end{example}

\section{Automorphisms}
\label{sec:automorph}

Recall that the \emph{(permutation) automorphism group} $ \Aut(C)$ of a code $C\subseteq \F^n$ is the set of all permutations $\sigma$ of $\{1,\dots,n\}$ such that $(c_{\sigma(1)},\dots,c_{\sigma(n)}) \in C$ for all $c=(c_1,\dots,c_n) \in C$. Evidently, $\Aut(C)$ is a subgroup of the symmetric group $S_n$ on $\{1,\dots,n\}$. In this section, we shall observe that the result 
stated in the introduction about the automorphism groups of affine Grassmann codes being large can be extended a little further.


For any $B \in \GL_{\ell}(\F)$, $A \in \GL_{\ell'}(\F)$, and ${\bf u} \in M_{\ell \times \ell'}(\F)$, define
$$
\psi_{{\bf u},A,B}: \Adelta \to \Adelta
$$
to be the affine transformation given by
$$
\psi_{{\bf u},A,B}(P)=BPA^{-1}+{\bf u} \quad \mbox{ for }  P=(p_{ij})_{1\le i \le \ell, \; 1\le j \le \ell'}\in \Adelta,
$$
It is clear that the transformation
$\psi_{{\bf u},A,B}$ gives a bijection of 
$\Adelta  =\{P_1, \dots , P_n\}$ onto itself, and hence 
there is a unique permutation $\sigma$ of $\{1,\dots,n\}$ 
such that
$$
\left(\psi_{{\bf u},A,B}(P_1), \dots , \psi_{{\bf u},A,B}(P_n)\right) = \left(P_{\sigma(1)}, \dots ,P_{\sigma(n)}\right).
$$
We shall denote this permutation $\sigma$ by $\sigma_{{\bf u},A,B}$ and for any $c=(c_1, \dots , c_n)\in \F^n$, we will often write
$\sigma_{{\bf u},A,B}(c)$ for the $n$-tuple $(c_{\sigma(1)},\dots,c_{\sigma(n)})$.

\begin{lemma}\label{lem:sigmauAB}
Let $r$ be a nonnegative integer $\le \ell$ and let  $B \in \GL_{\ell}(\F)$, $A \in \GL_{\ell'}(\F)$, and ${\bf u} \in M_{\ell \times \ell'}(\F)$. Then $\sigma_{{\bf u},A,B} \in \Aut\left(\Cr\right)$.
\end{lemma}
\begin{proof}
From Lemma 18 of \cite{BGH} and with its proof, we know that if $f= f(X)$ is in ${\mathcal F}(\ell,m;r)$, then $f(XA^{-1}+{\bf u})\in {\mathcal F}(\ell,m;r)$.
Now consider the product $BX$ and let $s$ be any integer with $0\le s \le r$. Observe that
any $s\times s$ minor of $BX$ is of the form $\det(B_sX^s)$, where $B_s$ is a $s\times \ell$ submatrix of $B$ and $X^s$ is a $\ell\times s$ submatrix of $X$. Hence by the Cauchy-Binet formula (cf. \cite[Lemma 10]{BGH}), every $s\times s$ minor of $BX$ is a $\F$-linear combination of $s\times s$ minors of $X$. Consequently, 
if $f \in {\mathcal F}(\ell,m;r)$, then $f(BXA^{-1}+{\bf u})\in {\mathcal F}(\ell,m;r)$. Moreover,
$$
\sigma_{{\bf u},A,B}\left(\Ev(f) \right)= \big(f(\psi_{{\bf u},A,B}(P))\big)_{P\in \Adelta}=\Ev\left(f(BXA^{-1}+{\bf u})\right).
$$
It follows that $\sigma_{{\bf u},A,B} \in \Aut(C)$, where $C = \Cr = \Ev \left({\mathcal F}(\ell,m;r)\right)$.
\end{proof}

Notice that $\psi_{\mathbf{0}, I_{\ell'},I_{\ell}}$ is the identity transformation of $\Aff^{\delta}$, where $\mathbf{0}$ denotes the zero matrix in
$M_{\ell \times \ell'}(\F)$ and $I_{\ell'}$ (resp: $I_{\ell}$) denotes the $\ell' \times \ell'$ (resp: $\ell \times \ell$) identity matrix over $\F$.   Moreover, given any $A, A' \in \GL_{\ell'}(\F)$, $B, B' \in \GL_{\ell}(\F)$, and
${\bf u}, \mathbf{v} \in M_{\ell \times \ell'}(\F)$, we have
\begin{equation}\label{eq:psiAB}
\psi_{{\mathbf{u}},A,B}\circ\psi_{{\mathbf{v}},A',B'}= \psi_{\mathbf{w},AA', BB'}
\quad \mbox{and} \quad \psi_{{\mathbf{u}},A,B}^{-1} = \psi_{\mathbf{u'},A^{-1}, B^{-1}},
\end{equation}
where $\mathbf{w}:= B{\mathbf{v}}A^{-1}+{\mathbf{u}}$ and $\mathbf{u'} = -B^{-1}\mathbf{u}A$.
%
It follows that 
$$
{\mathfrak H}(\ell,m):= \{\psi_{{\mathbf{u}},A, B} \, : \, A \in \GL_{\ell'}(\F), \; B\in \GL_{\ell}(\F), \mbox{ and }  {\mathbf{u}}\in M_{\ell\times\ell'}(\F)\}
$$
is a group with respect to composition of maps. 
We determine the group structure of ${\mathfrak H}(\ell,m)$ in the following result, which is an analogue of Proposition 20 of \cite{BGH}.

\begin{proposition}\label{prop:groupstructure}
Let $\Gamma(\ell,\ell')$ denote the factor group $G/Z$, where  $G$ is the direct product
$\GL_{\ell}(\F)\times \GL_{\ell'}(\F)$ and $Z$ is the normal subgroup of $G$ given by $\{\left(\lambda I_{\ell},\; \lambda I_{\ell'}\right) : \lambda\in \F^*\}$. Then
as a group ${\mathfrak H}(\ell,m)$ is isomorphic to the semidirect product $M_{\ell\times\ell'}(\F) \rtimes_{\theta} \Gamma(\ell,\ell')$, where $\theta:  \Gamma(\ell,\ell') \to Aut(M_{\ell\times\ell'}(\F))$ is a group homomorphism defined by $\theta\left((B,A)Z\right)(\mathbf{u}):=B\mathbf{u}A^{-1}$.
\end{proposition}

\begin{proof}
It is easy to check that $\theta$ is well-defined and that it is a group homomorphism.
Let $\eta: M_{\ell\times\ell'}(\F) \rtimes_{\theta} \Gamma(\ell,\ell') \to {\mathfrak H}(\ell,m)$ be the map given by
$\left({\mathbf{u}},\; (B,A)Z\right)\mapsto \psi_{{\mathbf{u}},A,B}$. Clearly, $\eta$ is well-defined and surjective.
Moreover, from \eqref{eq:psiAB} it is readily seen that $\eta$ is a group homomorphism. Finally, suppose $\left({\mathbf{u}},\; (B,A)Z\right)$ is in the
kernel of $\eta$ 
for some $B \in \GL_{\ell}(\F)$, $A \in \GL_{\ell'}(\F)$, and ${\bf u} \in M_{\ell \times \ell'}(\F)$. Then
\begin{equation}
\label{eq:bpainv}
BPA^{-1} + {\mathbf{u}} = P \quad \text{ for all } P \in \Adelta. 
\end{equation}
Taking $P$ to be the zero matrix in \eqref{eq:bpainv}, we obtain $\mathbf{u}=\mathbf{0}$. Next, write $B=(b_{ij})$ and $A^{-1}=(a'_{ij})$ and
let us fix any $r,s\in \Z$ with $1\le r\le \ell$ and $1\le s\le \ell'$. Taking $P$ to be the $\ell \times \ell'$ matrix $E_{rs}$, with $1$ in $(r,s)^{\rm th}$ spot and $0$ elsewhere, in \eqref{eq:bpainv}, we obtain
\begin{equation}
\label{eq:birasj}
b_{ir}\,a'_{sj} = \begin{cases} 1 & \text{ if } (i,j)=(r,s), \\ 0 & \text{ if } (i,j)\ne (r,s), \end{cases} \quad \text{ for } 1\le i \le \ell \text{ and } 1\le j \le \ell'.
\end{equation}
In particular, $b_{rr}\ne 0$ and $a'_{ss}\ne 0$. 
Now taking $j=s$ in \eqref{eq:birasj}, we obtain $b_{ir}=0$ for $i\ne r$. Likewise, $a'_{sj}=0$ for $j\ne s$. It follows that $B$ and $A^{-1}$ are diagonal matrices. Furthermore, thanks to \eqref{eq:birasj}, we have $b_{11}a'_{11}=\dots = b_{\ell\ell}a'_{11}=1$  and $b_{11}a'_{11}=\dots = b_{11}a'_{\ell'\ell'}=1$,
and therefore 
$B=\lambda I_{\ell}$ and $A^{-1}=\lambda^{-1}I_{\ell'}$ for some $\lambda\in \F^*$. This shows that the coset of $(B, A)$ in $G/Z$ is the identity element. Thus $\eta$ is an isomorphism.
\end{proof}

It may be noted that $C^\mathbb{A}(\ell,m;0)$ is a one-dimensional code of length $n=q^{\delta}$ spanned by $(1,1,\dots, 1)$ and thus its automorphism group is the full symmetric group $S_n$. For affine Grassmann codes of level $r\ge 1$, one has the following partial result, which extends Theorem 21 of \cite{BGH}.

\begin{theorem}
\label{thm:autom}
Let $r$ be a positive integer $\le \ell$.
Then the  automorphism group of  $\Cr$ contains a subgroup isomorphic to ${\mathfrak H}(\ell,m)$. In particular, 
\begin{equation}
\label{eq:Autsize}
\# \Aut\left(\Cr\right) \ge \frac{q^{\delta}}{q-1} \left(\prod_{i=0}^{\ell-1}(q^\ell-q^i)\right)\left(\prod_{j=0}^{\ell'-1}(q^{\ell'}-q^j)\right).
\end{equation} 
\end{theorem}
\begin{proof}
In view of Lemma \ref{lem:sigmauAB}, 
$\psi_{{\mathbf{u}},A,B} \mapsto \sigma_{{\mathbf{u}},A,B}$ gives a natural map from ${\mathfrak H}(\ell,m)$
into $\Aut\left(\Cr\right)$. It is readily seen that this map is a group homomorphism. 
So it suffices to show that this homomorphism is injective. To this end, suppose $\sigma_{{\mathbf{u}},A,B}$ is the identity permutation for some
$B \in \GL_{\ell}(\F)$, $A \in \GL_{\ell'}(\F)$, and  ${\mathbf{u}}\in M_{\ell\times\ell'}(\F)$.
Then $\sigma_{{\mathbf{u}},A,B}(\Ev(f))=\Ev(f)$ for all $f \in {\mathcal F}(\ell,m)$, i.e.,
$$
f(BPA^{-1}+{\mathbf{u}})=f(P) \quad \mbox{ for all $f \in {\mathcal F}(\ell,m)$ and all $P \in \Adelta$}.
$$
By letting $f$ vary over all possible $1\times 1$ minors, 
we see that
\eqref{eq:bpainv} holds. Hence $\psi_{{\mathbf{u}},A,B}$ is the identity transformation of
$\Aff^{\delta}$. Finally, \eqref{eq:Autsize} follows from Proposition \ref{prop:groupstructure}.
\end{proof}

\begin{remark}
\label{rem:transposeaut}
It may be tempting to believe that $\Aut\left(\Cr\right)$ is isomorphic to ${\mathfrak H}(\ell,m)$ for any $r\ge 1$. But already when $r=1$, we know from part (iii)~of Proposition \ref{pro:RMbasics} that $\Aut\left(C^\mathbb{A}(\ell,m;1)\right) = \Aut\left(\RM(1, \delta)\right)\simeq \AGL_{\delta}(\F)\simeq \F^{\delta} \rtimes \GL_{\delta}(\F)$, and the latter is, in general, much larger that ${\mathfrak H}(\ell,m)$. Even when
$r=\ell=\ell'>1$, one can see  as follows that $\Aut\left(\Cr\right) = \Aut\left(\C\right)$ can be larger than ${\mathfrak H}(\ell,m)$. Consider the
permutation $\sigma$ of $S_n$ induced by the transpose map, i.e., $\sigma\in S_n$ such that $(P_1^T, \dots , P_n^T) =
\left(P_{\sigma(1)}, \dots ,P_{\sigma(n)}\right)$. It is clear that the minors of $X^T$ are minors of $X$,
and hence $\sigma$ is an automorphism of $\C$. If $\sigma$ were equal to $\sigma_{{\mathbf{u}},A,B}$ for some
$B \in \GL_{\ell}(\F)$, $A \in \GL_{\ell'}(\F)$, and  ${\mathbf{u}}\in M_{\ell\times\ell'}(\F)$, then as in the proof of Theorem \ref{thm:autom}, we obtain
$$
BPA^{-1} + {\mathbf{u}} = P^T \quad \text{ for all } P =(p_{ij})_{1\le i \le \ell, \; 1\le j \le \ell'}\in \Adelta.
$$
Taking $P$ to be the zero matrix, we conclude that ${\mathbf{u}}= {\mathbf{0}}$. Further, since linear polynomials are reduced and hence determined by the corresponding $\F$-valued function on $\Adelta$, we see that $BXA^{-1}=X^T$. In particular, writing $B=(b_{ij})$ and $A^{-1} = \left(a'_{ij}\right)$, we see that
$$
X_{ii} = \sum_{r=1}^{\ell}\sum_{s=1}^{\ell} b_{ir}X_{rs}a'_{si} \quad \text{ for } i=1, \dots , \ell.
$$
Consequently, for any $i,r,s\in\{1, \dots , \ell\}$, we obtain $b_{ii} a'_{ii} = 1$ and $b_{ir} a'_{si} = 0$ if $(r,s)\ne (i,i)$. This, in turn,  implies that $B$ and $A^{-1}$ are diagonal matrices. But then the $(i,j)^{\rm th}$ entry of $BXA^{-1}$ is $b_{ii}X_{ij}a'_{jj}$, which can not always be $X_{ji}$ since $\ell >1$. This shows that $\sigma$ does not belong the subgroup of $\Aut\left(\C\right)$ corresponding to ${\mathfrak H}(\ell,m)$.
At any rate, the complete determination of 
$\Aut(\C)$ and more generally, $\Aut\left(\Cr\right)$ for $1< r\le \ell$, remains an open question.
\end{remark}

\section{Duality}
\label{sec:dual}

In this section we shall explicitly determine the dual of any affine Grassmann code and compute its minimum distance. Let us begin by observing that the monomial
$$
\full := 
\prod_{(i,j)\in \rec} X_{ij}^{q-1} = \prod_{i=1}^{\ell} \prod_{j=1}^{\ell'} X_{ij}^{q-1}
$$
is reduced and that $\mu \in \Mon$ is a reduced monomial if and only if $\mu$ divides $\full$. We may refer to $\full$ as the \emph{full product}.
Note that for $0\le r\le \ell$,
$$
\dim \rp =\#\, \RMon = \sum_{s=0}^{\ell} \binom{\delta}{s} (q-1)^s = q^{\delta} = \text{length}\left(\Cr\right)
$$
and also that
\begin{equation}
\label{eq:dimCrp}
\dim \Cr^{\perp} = n - \dim \Cr = q^{\delta} - \sum_{i=0}^{r} {\binom{\ell}{i}}{\binom{\ell'}{i}}.
\end{equation}
The usual ``inner product'' on $\F^{q^\delta}$ corresponds to the symmetric bilinear form $\langle \;, \;\rangle$ on the $\Fq$-linear space
$\rp$ given by
$$
\left\langle f,g\right\rangle :=  \sum_{P\in \Ad} f(P)g(P) =   \sum_{P\in \Ad} fg(P) =  \sum_{P\in \Ad} \overline{fg}(P) ,  \quad \text{ for } f,g\in \rp.
$$
The dual of $\Cr$ corresponds, via the $\Fq$-linear isomorphism $f\mapsto \Ev(f)$ of $\rp \to \F^{q^\delta}$, to the subspace
\begin{equation}
\label{eq:dualinrp}
\left\{f\in \rp : \left\langle f,{\mathcal M}\right\rangle = 0 \text{ for all } {\mathcal M} \in \Delta(\ell,m;r)\right\}
\end{equation}
of $\rp$. We shall now proceed to determine an explicit $\F$-basis of this subspace. The first step is to recall the following
well-known result (cf. \cite[Lem. 1.6]{DGM}).

\begin{proposition}
\label{prop:muPzero}
Let $\mu \in \RMon$ be a reduced monomial. Then
$$
 \sum_{P\in \Ad} \mu (P) = \begin{cases} 0 & \text{ if }\, \mu \ne \full, \\ (-1)^{\delta} & \text{ if }\, \mu = \full.\end{cases}
$$
\end{proposition}
We have noted in Remark \ref{multred} that $\overline{fg}$ need not be equal to
$\bar f \bar g$ for arbitrary $f,g\in \Fq[X]$. The following useful lemma shows what is the best that we can do in a special case.

\begin{lemma}
\label{lem:munuF}
Let $\mu, \nu \in \RMon$ be such that $\overline{\mu \nu} = \full$. Then there is a divisor $\nu'$ of $\nu$ such that $\mu\nu'=\full$.
Moreover, if $\nu$ is squarefree and if $q>2$, then 
$\mu\nu=\full$.
\end{lemma}

\begin{proof}
For $(i,j)\in \rec$, 
let $\alpha_{ij} = \deg_{X_{ij}} \mu$ and  $\beta_{ij} = \deg_{X_{ij}} \nu$. Since
$\mu$ and $\nu$ are reduced, it follows that
$$
\deg_{X_{ij}} \overline{\mu \nu} =  \begin{cases}  \alpha_{ij} + \beta_{ij} & \text{ if }\,  \alpha_{ij} + \beta_{ij} \le q-1, \\
 \alpha_{ij} + \beta_{ij} - q+ 1  & \text{ if }\,  \alpha_{ij} + \beta_{ij} \ge q.\end{cases}
$$
On the other hand, since $\overline{\mu \nu} = \full$, we see that $\alpha_{ij} + \beta_{ij} \ge q-1$ for all $(i,j)\in \rec$. 
Hence
$
\nu' :=  \prod_{(i,j)\in \rec} 
X_{ij}^{q-1-\alpha_{ij}}
$
is a divisor of $\nu$ and it clearly satisfies $\mu\nu'=\full$. Finally, suppose $\nu$ is squarefree and $q>2$, but $\nu'\ne \nu$. Then there is a variable $X_{ij}$ that divides $\nu$, but not 
$\nu'$. Now since $\mu\nu'=\full$, we see that 
$\deg_{X_{ij}}\mu = q-1$. But then $\deg_{X_{ij}}\mu\nu = q$, which contradicts the assumption that $\overline{\mu \nu} = \full$, since $q>2$.
\end{proof}

Given any nonnegative integer $r\le \ell$, define
$$
\FM:=\left\{\frac{\full}{t} \, : \, t\in \Term({\mathcal M}) \text{ for some } {\mathcal M}\in \Delta(\ell,m;r)\right\}.
$$
It is clear that elements of $\FM$ are reduced monomials; we shall refer to them as
\emph{forbidden monomials} with respect to the affine Grassmann code of level $r$. This terminology is justified by the following result.

\begin{lemma}
\label{lem:NFM}
Let $r$ be a nonnegative integer $\le \ell$ and let
$\mu\in \RMon$ be such that $\mu\not\in\FM$. Then $\Ev(\mu)\in \Cr^{\perp}$.
\end{lemma}

\begin{proof}
Let ${\mathcal M}\in \Delta(\ell,m;r)$ and let $t\in \Term({\mathcal M})$. Now $\overline{\mu t}$ is reduced and if it were equal to $\full$, then by Lemma \ref{lem:munuF},
$\mu =\full/t'$ for some divisor $t'$ of $t$. But this contradicts the assumption that $\mu\not\in\FM$ because the divisor of a term of a
minor in $\Delta(\ell,m;r)$ is also a term of a minor in $\Delta(\ell,m;r)$. Thus, in view of Proposition~\ref{prop:muPzero}, we obtain
$\left\langle \mu , t \right\rangle = \sum_{P\in \Ad} \overline{\mu t}(P)=0$. Consequently, $\mu$ is in the subspace of $\rp$ given by \eqref{eq:dualinrp}, and so $\Ev(\mu)\in \Cr^{\perp}$.
\end{proof}

Already, we have enough information to show that affine Grassmann codes are almost always self-orthogonal. More precisely, we have the following.

\begin{theorem}
\label{selforth}
Let $r$ be a nonnegative integer $\le \ell$. Then the affine Grassmann code $\Cr$ of level $r$ is self-orthogonal if and only if $(\ell,m;r;q)$ is different from $(1,2;1;2)$, $(1,2;1;3)$ and $(1,3;1;2)$.
\end{theorem}

\begin{proof}
First, if $r=0$, then $\Cr$ is the one-dimensional code spanned by the all $1$-vector $(1, 1, \dots , 1)$ in $\Fq^{q^\delta}$ and this is clearly self-orthogonal. Now suppose $r\ge 1$.
Observe that if $\mu \in \FM$ is any forbidden monomial, then
$$
\deg \mu \ge \deg \full - r  \ge \deg \full - \ell  = \left[(q-1)\ell'-1\right] \ell  \ge \left[(q-1)\ell'-1\right] r. 
$$
In particular, if $(q-1)\ell' > 2$, then 
no reduced monomial of degree $\le r$ is forbidden.
On the other hand, 
$\Cr$ is spanned by the evaluations of minors of size $\le r$, which, in turn, are $\F$-linear combinations of 
reduced monomials of degree $\le r$.
Hence by Lemma \ref{lem:NFM}, we can conclude that 
$\Cr \subseteq \Cr^{\perp}$ when $(q-1)\ell' > 2$.
Now suppose $(q-1)\ell'\le 2$. Since $1\le r\le\ell\le \ell'$, the only possible values of
$(\ell,\ell';r;q)$ are $(1,1;1;2)$, $(1,1;1;3)$, $(1,2;1;2)$, and $(2,2;2;2)$. For the first $3$ values, one finds $\dim \Cr > \dim \Cr^{\perp}$,
and hence $\Cr$ is not self-orthogonal in these cases. When $r=\ell=\ell'=q=2$, the code $\Cr$ is spanned by the evaluations of
$1, X_{11}, X_{12},X_{21}, X_{22}$ and the 
minor $B := X_{11}X_{22}+X_{12}X_{21}$. 
The first $5$ are non-forbidden reduced monomials; hence by Lemma \ref{lem:NFM}, they are in $\Cr^{\perp}$.
A direct verification shows that $\langle B,B\rangle =0$, since $q=2$. Thus $\Cr$ is self-orthogonal when $r=\ell=\ell'=q=2$.
\end{proof}

Although the non-forbidden monomials give rise to linearly independent elements of the dual of an affine Grassmann code, they fail to span it. To extend these to a basis, one needs to add certain binomials such as the polynomial $B$ in the proof of Theorem \ref{selforth}.
A general definition of these binomials is given below. 

First, let us introduce some notation, which will be useful in the sequel. For any nonnegative integer $r\le \ell$,  denote, as usual, by $\Sr$ the set of all permutations  of $\{1, \dots , r\}$. Further, given any $r\times r$ minor ${\mathcal M}$ of $X$ and any $\sigma \in \Sr$, 
denote by $\tsM$ the \emph{signed term} of ${\mathcal M}$ corresponding to the permutation $\sigma$. For example, ${{\mathsf{t}}_{\sigma}({\mathcal{L}_r})}= \sgn(\sigma)X_{1\sigma(1)}\cdots X_{r\sigma(r)}$, where
${\mathcal L}_r$ is the $r^{\rm th}$ leading principal minor of $X$.
We will denote by $\epsilon$ the identity permutation and, by abuse of language, regard it as an element of $\Sr$ 
for every nonnegative integer $r$. In particular, for any minor ${\mathcal M}$ of $X$, the corresponding signed term
$\teM$ is precisely the product of the variables on the principal diagonal of the submatrix corresponding to ${\mathcal M}$.
Define
$$
\BtM := \frac{\full}{\teM} - \frac{\full}{\tsM}  \quad \text{ for ${\mathcal M}\in \Delta_r(\ell,m)$ and } \sigma\in \Sr.
$$
Clearly, $\BtM=0$ if $\sigma=\epsilon$ and in particular, if $r\le 1$. If $r\ge 2$ and if $\sigma$ is a non-identity permutation, then $\BtM$ is a reduced polynomial with exactly two terms, each of which is a forbidden monomial up to a sign. We may refer to $\BtM$ as the \emph{binomial} corresponding to the minor ${\mathcal M}$ and the 
permutation~$\sigma$.

\begin{lemma}
Let $i,r$ be integers such that $0\le i \le r \le \ell$, and let
${\mathcal M} \in \Delta_i(\ell,m)$ and $\sigma\in \Si$. 
Then $\Ev\left(\BtM\right) \in \Cr^{\perp}$. 
\end{lemma}

\begin{proof}
Clearly, it suffices to show that $\langle \BtM, \, {\mathcal N}\rangle = 0$ for all ${\mathcal N}\in \Delta(\ell,m;r)$.
So let us fix some $j\times j$ minor ${\mathcal N}$ of $X$, where $j\le r$.
Also let $\pi$ denote a permutation of $\{1,\dots , j\}$.
We will distinguish two cases.

{\bf Case 1:} $q>2$. 
Since $\sgn(\pi)\tpM$ is a squarefree monomial, it follows from  Lemma \ref{lem:munuF} that
$\overline{\left({\full}/{\tsM}\right)\tpM} = \pm \full$ only when
$\tsM=\pm \tpM$, which, in turn, is possible only when 
$i=j$, $ {\mathcal M}={\mathcal N}$, and $\sigma=\pi$.
Consequently, in view of Proposition \ref{prop:muPzero}, we see that
$\langle \BtM, \, {\mathcal N}\rangle = 0$ if $\mathcal N \neq \mathcal M$, whereas
$$
\langle \BtM, \, {\mathcal M}\rangle = \sum_{\pi\in \Sj} \langle \BtM, \, {\mathsf{t}}_{\pi}({\mathcal{M}}) \rangle
= \sgn(\epsilon)^2 (-1)^{\delta} - \sgn(\sigma)^2 (-1)^{\delta} =0.
$$

{\bf Case 2:} $q=2$. In this case it follows from  Lemma \ref{lem:munuF} that
$\overline{\left({\full}/{\tsM}\right)\tpM} =  \full$ only when
$\tsM$ divides $\tpM$. Further, if $Y$ denotes the $j\times j$ submatrix of $X$ corresponding to the minor ${\mathcal N}$, then it is readily seen that $\tsM$ divides $\tpM$ if and only if
$i\le j$, $\, {\mathcal M}=\det Y'$, 
and $\sigma= \pi '$, where $Y'$ is an $i\times i$ submatrix of  $Y$ and $\pi'$ is the restriction to $\pi$ to $\{1, \dots , i\}$.
Consequently, in view of Proposition \ref{prop:muPzero}, we see that $\langle {\full}/{\tsM}, \, \tpM \rangle =1$ for precisely $(j-i)!$ permutations $\pi\in \Sj$ obtained by extending $\sigma$ 
to $\{1, \dots , j\}$ by permuting  $i+1, \dots , j$ randomly. It follows that
$$
\langle \BtM, \, {\mathcal N}\rangle = 2 (j-i)! = 0.
$$
This completes the proof.
Consequently, in view of Proposition \ref{prop:muPzero}, we obtain
$\langle \BtM, \, {\mathcal N}\rangle = 0$ if $\mathcal N \neq \mathcal M$, while
$$
\langle \BtM, \, {\mathcal M}\rangle = \sum_{\pi\in \Sj} \langle \BtM, \, {\mathsf{t}}_{\pi}({\mathcal{M}}) \rangle
= \sgn(\epsilon)^2 (-1)^{\delta} - \sgn(\sigma)^2 (-1)^{\delta} =0.
$$
\end{proof}

We are now ready to describe an explicit basis for $\Cr^{\perp}$. In fact, this is given by the non-forbidden monomials and the binomials. More precisely, for a nonnegative integer $r\le \ell$, we let
$$
\basis := \left(\RMon \setminus \FM \right)\cup \left( \bigcup_{i=0}^r \left\{\BtM : {\mathcal M} \in \Delta_i(\ell,m) \text{ and } \sigma\in \Si^*\right\}
\right),
$$
where $\Si^*:=\Si\setminus\{\epsilon\}$ is the set of non-identity permutations of $\{1, \dots , i\}$; also let
$$
\fstar := \text{$\Fq$-linear span of }\basis .
$$
Note that $\fstar$ is a subspace of $\rp$ and, in particular, it is $\Fq$-isomorphic to its image in $\Fq^{q^{\delta}}$ under the evaluation map.
Now we have the following explicit description of the dual of an affine Grassmann code of any given level.

\begin{theorem} 
\label{dualdescription}
$\Cr^{\perp} = \Ev\left(\fstar\right)$ for $0\le r \le \ell$.
\end{theorem}

\begin{proof}
Fix a nonnegative integer $r\le \ell$. Let us first show that the elements of $\basis$ are linearly independent. Suppose
$$
\sum_{\mu} a_{\mu} \mu + \sum_{i=0}^r \sum_{\mathcal{M}\in \Delta_i(\ell,m)} \sum_{\sigma\in \Si^*} b_{i,\sigma, \mathcal{M}} \, \BtM = 0,
$$
for some $ a_{\mu}, \, b_{i,\sigma, \mathcal{M}} \in \F$, where $\mu$ varies over $\RMon \setminus \FM $. Then
\begin{equation}
\label{newlc}
\sum_{\mu} a_{\mu} \mu + \sum_{i=2}^r \sum_{\mathcal{M}\in \Delta_i(\ell,m)} \sum_{\sigma\in \Si} b_{i,\sigma, \mathcal{M}} \, \frac{\full}{\tsM} = 0,
\end{equation}
where, for $2\le i\le r$ and $\mathcal{M}\in \Delta_i(\ell,m)$, we have put
$b_{i,\epsilon, \mathcal{M}}:=- \sum_{\sigma\in \Si^*} b_{i,\sigma, \mathcal{M}}$. Now observe that \eqref{newlc} is a linear combination of distinct monomials. 
Hence we must have $a_{\mu}=0$ and $b_{i,\sigma, \mathcal{M}}=0$ for all relevant parameters $\mu, i,\sigma,$ and $\mathcal{M}$.

To complete the proof, it suffices to show that the cardinality of $\basis$ coincides with the dimension of the subspace of $\rp$ given by
\eqref{eq:dualinrp}. 
To this end, let us first note that a forbidden monomial is completely determined by an $i\times i$ minor of $X$
and by one of its $i!$ terms. Since $X$ has exactly $\binom{\ell}{i}\binom{\ell'}{i}$ minors, we see that
$$
\# \RMon \setminus \FM = q^{\delta} - \sum_{i=0}^r i! \binom{\ell}{i}\binom{\ell'}{i}.
$$
On the other hand, the binomials are determined by an $i\times i$ minor of $X$ and a non-identity permutation of $\{1,\dots , i\}$. Thus,
$$
\# \bigcup_{i=0}^r \left\{\BtM : {\mathcal M} \in \Delta_i(\ell,m) \text{ and } \sigma\in \Si^*\right\} = \sum_{i=0}^r (i!-1) \binom{\ell}{i}\binom{\ell'}{i}.
$$
Combining the last two equations, we see that $\# \basis$ is the expression on the right in \eqref{eq:dimCrp}, as desired.
\end{proof}

We shall now proceed to determine the minimum distance of the dual of an affine Grassmann code.
As a warm-up, it may be noted that the Singleton bound shows already that for any nonnegative integer $r\le \ell$,
$$
d\left(\Cr^{\perp}\right) \le 1 + \sum_{i=0}^r  \binom{\ell}{i}\binom{\ell'}{i} \le  1+ \binom{m}{\ell}.
$$
This indicates that the minimum distance is rather small and it does not grow with $q$. In the trivial case $r=0$, we obtain $2$ as an
upper bound, and it is readily seen that this is attained. Indeed, ${C^\mathbb{A}(\ell,m;0)}$ is the one-dimensional code of length $q^{\delta}$
spanned by $(1,1,\dots , 1)$ and its dual contains no codeword of weight $1$. Another trivial case is when $q=2$ and $r=\ell=\ell'=1$. In this
case, ${C^\mathbb{A}(1,2;1)} = {\mathbb F}_2^2$, while ${C^\mathbb{A}(1,2;1)}^{\perp} = \{\mathbf{0}\}$. Barring these, it will be seen below that
the minimum distance is always $3$ or $4$.

\begin{theorem}
\label{mindistCrdual}
Let $r$ be a positive integer $\le \ell$. Then the minimum distance of the $q$-ary code $\Cr^{\perp}$ is given by
$$
d\left(\Cr^{\perp}\right) = \begin{cases} 3 & \text{ if } \; q > 2, \\ 4 & \text{ if } \; q = 2 \, \text{ and } \, \ell'>1. \end{cases}
$$
Moreover, if $q>2$ and if $a_1, a_2$ are any distinct elements of $\Fq^*$, then $\Ev\left(g_{a_1, a_2}\right)$ is a minimum-weight codeword of $\Cr^{\perp}$, where
\begin{equation}
\label{eq:g}
g_{a_1, a_2}:= \frac{1}{\left(X_{\ell\ell'}-a_1\right) \left(X_{\ell\ell'}-a_2\right)} \prod_{(i,j)\in \rec} \left(X_{ij}^{q-1}-1\right).
\end{equation}
On the other hand, if 
$q=2$ and $\ell'>1$, then there are distinct $(i_1, j_1), (i_2, j_2)\in \rec$ such that $i_1=i_2$ or $j_1=j_2$, and moreover for any such
$(i_1, j_1), (i_2, j_2)$, if we let
\begin{equation}
\label{eq:h}
h := \frac{\full}{X_{i_1j_1}X_{i_2j_2}}. 
\end{equation}
then $\Ev(h)$ is a minimum-weight codeword of $\Cr^{\perp}$.
\end{theorem}

\begin{proof}
Let us assume that either $q>2$ or that $q=2$ and $\ell'>1$. This ensures that $\delta(q-1) -2\ge 0$. 
Now, observe that 
every element of $\basis$ is either a reduced monomial of degree $\le \delta(q-1) -2$ or a difference of two reduced monomials of degree $\le \delta(q-1) -2$. 
Hence it follows from Theorem \ref{dualdescription}
that $\Cr^{\perp}$ is a subcode of the generalized Reed-Muller code $\RM\left(\delta(q-1)-2, \, \delta\right)$. 
Consequently, from part (ii) of Proposition \ref{pro:RMbasics}, we see that
$$
d\left(\Cr^{\perp}\right) \ge d\left(\RM\left(\delta(q-1)-2, \, \delta\right)\right)=\begin{cases} 3 & \text{ if } \; q > 2, \\ 4 & \text{ if } \; q = 2 \, \text{ and } \, \ell'>1. \end{cases}
$$
To complete the proof, it suffices to show that the evaluations of \eqref{eq:g} and \eqref{eq:h} give codewords of (Hamming) weight $3$ and $4$, respectively.

To begin with, suppose $q>2$ and let $a_1, a_2$ be any distinct elements of $\Fq^*$. Since $X^{q-1} - 1 = \prod_{a\in \Fq^*}(X-a)$, it is clear that $g_{a_1,a_2}$ defined by \eqref{eq:g} is in $\Fq[X]$ and is, in fact, a reduced polynomial. Moreover,
$\deg_{X_{\ell\ell'}}g_{a_1,a_2} \le q-3$.
On the other hand, since the terms of any minor are squarefree monomials, every forbidden monomial $\mu\in \FM$ must satisfy $\deg_{X_{ij}}\mu \ge q-2$ for all
$i=1, \dots , \ell$ and $j=1, \dots , \ell'$. It follows that $g_{a_1,a_2}$ is a $\Fq$-linear combination of non-forbidden reduced monomials and in particular, it is in $\fstar$. Moreover, if $P=(p_{ij})\in \Adelta$, then $g_{a_1,a_2}(P)\ne 0$ if and only if $p_{ij}=0$ for all $(i,j)\ne (\ell,\ell')$ and $p_{\ell\ell'}\in \{0, a_1, a_2\}$. Thus we conclude that $\Ev\left(g_{a_1,a_2}\right)$ is a codeword of $\Cr^{\perp}$ of weight $3$.

Next, suppose $q=2$ and $\ell'>1$. The existence of distinct $(i_1, j_1), (i_2, j_2)\in \rec$ such that $i_1=i_2$ or $j_1=j_2$ is obvious; for example, we can take $i_1=i_2=\ell$, $j_1=\ell'-1$  and $j_2=\ell'$. Moreover, for any such $(i_1, j_1), (i_2, j_2)\in \rec$, the monomial
$X_{i_1j_1}X_{i_2j_2}$ contains two variables from the same row or from the same column, and hence it can never be the term of any minor of $X$. Consequently, the reduced monomial $h$ defined by $\eqref{eq:h}$ is non-forbidden and $\Ev(h)$ is a codeword of $\Cr^{\perp}$. Furthermore, if $P=(p_{ij})\in \Adelta$, then $h(P)\ne 0$ if and only if $p_{ij}=1$ for all $(i,j)\in \rec$ 
different from $(i_1, j_1)$ and $(i_2, j_2)$.
Thus we conclude that $\Ev(h)$ is of weight $4$.
\end{proof}

\section{Generation by Minimum-weight Codewords}

In this section we will show that  the affine Grassmann codes as well as their duals have the property that the
codewords of minimum weight generate the code. 
%
The case of affine Grassmann codes is easy and in fact, it is shown below that the result holds more generally for affine Grassmann codes
of any level.

\begin{proposition}
\label{AffGC}
Let $r$ be a nonnegative integer $\le \ell$. Then the minimum-weight codewords of $\Cr$ generate $\Cr$.
\end{proposition}

\begin{proof}
The code $\Cr$ is generated by $\Ev(\mathcal{M})$ as $\mathcal{M}$ varies over the $i\times i$ minors of $X$ for $0\le i \le r$. We proceed by decreasing induction on $i$ ($0\le i \le r$) to show that $\Ev\left(\mathcal{M}\right)$ is in the $\Fq$-linear span of minimum-weight codewords of $\Cr$ for every $i\times i$ minor $\mathcal{M}$ of $X$.  To begin with, if $i=r$, then $\mathcal{M}$ is the $r^{\rm th}$ leading principal minor of $Y:=BXA$ for some permutation matrices $B \in \GL_{\ell}(\F)$ and $A \in \GL_{\ell'}(\F)$. Hence
Lemma \ref{lem:sigmauAB} shows that $\Ev\left(\mathcal{M}\right)$ differs from $\Ev\left(\Lr\right)$ by an automorphism of $\Cr$; consequently, by
Theorem \ref{thm:mindistClmr}, $\Ev\left(\mathcal{M}\right)$ is itself a minimum-weight codeword of $\Cr$.
Now, suppose $i < r$ and the result holds for the $(i+1)\times (i+1)$ minors of $X$. 
Let $a_1, \dots, a_i$ and $b_1, \dots , b_i$ denote, respectively, the row and column indices of $X$ corresponding to the $i\times i$ minor $\mathcal{M}$. Since $i < r$, we can choose a row index $\alpha$ distinct from $a_1, \dots , a_i$ and a column index $\beta$ distinct from $b_1, \dots , b_i$. Consider $X'=X+\mathbf{u}$, where $\mathbf{u}$ is the $\ell\times \ell'$ matrix whose $(\alpha, \beta)^{\rm th}$ entry is $1$ and all other entries are $0$. Let $\mathcal{N}$ (resp: $\mathcal{N}'$) be the $(i+1)\times (i+1)$ minor of $X$ (resp: $X'$) corresponding to the row indices $a_1, \dots , a_i, \alpha$ and column indices $b_1, \dots , b_i, \beta$. Observe that $\mathcal{M}=\mathcal{N}' - \mathcal{N}$.
From the induction hypothesis together with Lemma~\ref{lem:sigmauAB}, it follows that both $\Ev\left(\mathcal{N}\right)$ and $\Ev\left(\mathcal{N}'\right)$ are in the $\Fq$-linear span of minimum-weight codewords of $\Cr$ and therefore, so is $\Ev\left(\mathcal{M}\right)$.
\end{proof}

\begin{remark}
\label{rem:Grassmann}
Affine Grassmann codes are closely related to Grassmann codes, and this connection was explained in Section VII of \cite{BGH}.
We remark here that a result analogous to Proposition~\ref{AffGC} holds for Grassmann codes as well. To see this, it suffices to note that by a result of Nogin (see, e.g., \cite[Cor. 19]{GPP}), the minimum-weight codewords of the Grassmann code $C(\ell, m)$ correspond precisely to the decomposable elements in the exterior power $\wedge^{m-\ell}\Fqm$ and evidently, these decomposable elements span the corresponding function space
${\mathcal G}(\ell,m) = (\wedge^{\ell}\Fqm)^* \simeq \wedge^{m-\ell}\Fqm$.
\end{remark}

As indicated in the Introduction, an analogous result for the dual of $\Cr$ is not true, in general. However, the minimum-weight codewords of $\C^{\perp}$ do generate $\C^{\perp}$. In other words, a result analogous to Proposition \ref{AffGC} holds  for the duals of affine Grassmann codes (of
level $\ell$). This, in fact, seems much harder to prove and we will need a number of auxiliary results, which will be spread over the next three subsections. The first subsection contains lemmas of a general nature concerning generating sets and bases for certain spaces of polynomials. Next, we show that the evaluations of certain non-forbidden monomials with respect to $\Cr$ are generated by the minimum-weight codewords. Finally, the binomials in $\basis$ are dealt with in the last subsection, where we conclude with the main result of this section. Wherever possible, we will consider affine Grassmann codes of an arbitrary level so as to make it clear what works in general and what goes wrong when $r< \ell$ as opposed to $r=\ell$.

Before proceeding with generalities, and as a warm-up, let us consider the case of $r=0$. Here
 ${C^\mathbb{A}(\ell,m;0)}$ is the one-dimensional code of length $n:=q^{\delta}$
spanned by $(1,1,\dots , 1)$ and ${C^\mathbb{A}(\ell,m;0)}^{\perp}= \{(c_1, \dots , c_n)\in \F^n: c_1+\dots + c_n= 0\}$. Thus,
if $\{\mathbf{e}_1, \dots, \mathbf{e}_n\}$ denotes the standard basis of $\F^n$, then
$\mathbf{e}_1 - \mathbf{e}_2, \dots , \mathbf{e}_1 - \mathbf{e}_n$ are minimum-weight codewords and these clearly generate
${C^\mathbb{A}(\ell,m;0)}^{\perp}$.

\subsection{Generators and Bases}
\label{subsec1}
Let $T$ be an indeterminate over $\Fq$ and $d$ a nonnegative integer. Denote by $\Fq[T]$ the space of polynomials in $T$ with coefficients in $\Fq$,
and by $\fqd$ the subspace of polynomials in $\Fq[T]$ of degree $\le d$. 
Also, denote by $M_d[T]$ the set of monic polynomials in $\Fq[T]$ of degree $d$ having $d$ distinct roots in $\Fq$.


\begin{lemma}
\label{lem:mdt}
Assume that $d< q$. 
Then $M_d[T]$ spans $\fqd$.
\end{lemma}

\begin{proof}
Since $d<q$, we can choose distinct elements $a_1,\dots,a_{d+1}$ from $\F$. The $d+1$ polynomials
$$
p_i(T):=\mathop{\prod_{1\le j \le d+1}}_{j \ne i} (T-a_j) \quad \text{ for } i=1, \dots , d+1,
$$
are elements of $M_d[T]$. Moreover, they are linearly independent, since a relation $\sum_j \beta_j p_j(T)=0$ implies, after substituting $a_i$ for $T$, that $\beta_i=0$ for  $i=1, \dots , d+1$. Since $M_d[T] \subset \F[T]_{\le d}$ and $\dim_{\F} \F[T]_{\le d}=d+1$, the lemma follows.
\end{proof}

\begin{corollary}
\label{cor:Mqminus1T}
$M_{q-1}[T] =\{(T-{\alpha})^{q-1}-1 : {\alpha} \in \F\}$ and moreover, $M_{q-1}[T]$ is a $\Fq$-basis of $\Fq[T]_{\le q-1}$.
\end{corollary}

\begin{proof}
For each ${\alpha}\in \Fq$, the polynomial $(T-{\alpha})^{q-1}-1$ is clearly monic of degree $q-1$ and its roots are precisely the elements ${\beta}\in \Fq$ with ${\beta}\ne {\alpha}$. It follows that $M_{q-1}[T] =\{(T-{\alpha})^{q-1}-1 : {\alpha} \in \F\}$. In particular, $\# M_{q-1}[T] = q = \dim \Fq[T]_{\le q-1}$. Hence by Lemma \ref{lem:mdt},
$M_{q-1}[T]$ is a $\Fq$-basis of $\Fq[T]_{\le q-1}$.
\end{proof}

\begin{remark}
\label{rem:MdT}
For $0\le d < q$, the set $M_d[T]$ is a basis of $\F[T]_{\le d}$ if and only if $d=0$ or $d=q-1$. The case $d=0$ is trivial whereas $d=q-1$ was noted above. For the converse, it suffices to observe that $\# M_d[T] = {q\choose d} \ge q > d+1$ when $1\le d < q-1$.
In general, for $0\le d < q$, upon letting $e=q-1-d$, one can write
\begin{equation}
\label{eq:MdT}
M_{d}[T] =\left\{\frac{(T-{\alpha})^{q-1}-1}{(T-a_1) \cdots \left(T-a_{e}\right)} : \alpha, a_1, \dots , a_{e} \text{ distinct elements of } \F
  \right\}.
\end{equation}
This representation is particularly useful for large values of $d$. It may be noted, however, that for 
a given polynomial in $M_d[T]$, 
the corresponding $\alpha \in \Fq$ and the $e$-element subset $\{a_1, \dots , a_e\}$
of $\F\setminus \{\alpha\}$ 
is not unique.
\end{remark}

We now derive a multivariable analogue of Lemma \ref{lem:mdt}. To this end, let $s$ be a positive integer and $T_1, \dots , T_s$ independent indeterminates over $\Fq$, and let $d_1, \dots , d_s$ be nonnegative integers. Denote by $\Fq[T_1, \dots , T_s]$ the space of polynomials in
$T_1, \dots , T_s$ with coefficients in $\Fq$ and by $\fqds$ the subspace of polynomials $f\in \Fq[T_1, \dots , T_s]$ with
$\deg_{X_j} f \le d_j$ for $j=1, \dots , s$. Also, let 
$$
\mdts =\left\{\prod_{i=1}^s f_i(T_i) 
\; : \; f_i(T_i)\in M_{d_i}[T_i] \text{ for } i=1, \dots , s\right\}.
$$

\begin{lemma}
\label{lem:mdts}
Assume that $d_i< q$ for $i=1, \dots , s$. Then $\mdts$ spans 
$\fqds$.
\end{lemma}

\begin{proof}
$\fqds$ is generated by monomials of the form $T_1^{e_1}\cdots T_s^{e_s}$ with $0\le e_i \le d_i$
for $i=1, \dots , s$, and by Lemma \ref{lem:mdt}, each factor $T_i^{e_i}$ of such a monomial is a $\Fq$-linear combination of
elements of $M_{d_i}[T_i]$.
\end{proof}

As in Remark \ref{rem:MdT}, it may be noted that $\mdts$ is a basis of $\fqds$ if and only if  $d_1=\cdots = d_s =0$ or
$d_1=\cdots = d_s =q-1$.  In particular, $M_{(q-1, \dots , q-1)}[T_1, \dots , T_s]$
is a basis of the space $\Fq[T_1, \dots , T_s]_{\le (q-1, \dots , q-1)}$ of all reduced polynomials in $T_1, \dots, T_s$ with coefficients in $\Fq$.
The following lemma gives several other bases for this space. As in Section \ref{sec:prelim}, for any $f\in \Fq[T_1, \dots , T_s]$, we denote
by $\overline{f}$ the reduced polynomial in $\Fq[T_1, \dots , T_s]$ corresponding to $f$. Note that if $L\in \Fq[T_1, \dots , T_s]$ is a homogeneous linear polynomial, i.e., if $L = a_1T_1 + \cdots + a_sT_s$ for some $a_1, \dots , a_s\in \F$, then $\overline{L} = L$. In particular, $L$ can be identified with the functional $\Fq^s \to \Fq$ that maps $w=(w_1, \dots , w_s)\in \F^s$ to $L(w)= a_1w_1 + \cdots + a_sw_s$.
%

\begin{lemma}\label{lem:basis2}
Let $\{L_i : 1 \le i \le s \} \subseteq \F[T_1,\dots,T_s]$ be a set of $s$
linearly independent homogeneous linear polynomials.
Then the set
$$
{\mathfrak L} := \left\{\overline{L_1^{e_1}\cdots L_s^{e_s}} : 0 \le e_i \le q-1 \text{ for }  i=1, \dots , s \right\}
$$
is a basis of $\F[T_1,\dots,T_s]_{\le (q-1,\dots,q-1)}$.
\end{lemma}
\begin{proof}
Since the linear polynomials $L_1,\dots,L_s$ are linearly independent, the map given by $w\mapsto \left(L_1(w), \dots , L_s(w)\right)$ is a
$\Fq$-linear isomorphism of $\Fq^s$ onto $\F^s$. Hence given any $v \in \F^s$, there exists  $w_v \in \F^s$ such that
$\left(L_1(w_v),\dots,L_s(w_v)\right)=v$. Now let a relation
$\sum_{e_1,\dots,e_s}\alpha_{e_1,\dots,e_s}\overline{L_1^{e_1}\cdots L_s^{e_s}}=0$ be given, where 
$\alpha_{e_1,\dots,e_s} \in F$ for all $e_1, \dots e_s$ (with $0 \le e_i \le q-1 \text{ for }  i=1, \dots , s$).
Evaluating the given relation
at $w_v$, we find 
$\sum_{e_1,\dots,e_s}\alpha_{e_1,\dots,e_s} \alpha_{e_1,\dots,e_s}v_1^{e_1}\cdots v_s^{e_s} =0$.
Consequently, the polynomial $\sum_{e_1,\dots,e_s}\alpha_{e_1,\dots,e_s} \alpha_{e_1,\dots,e_s}T_1^{e_1}\cdots T_s^{e_s}$ vanishes at all
points of $\F^s$. Since $0 \le e_i \le q-1$ for $i=1, \dots , s$, this is only possible if $\alpha_{e_1,\dots,e_s}=0$ for all $e_1,\dots,e_s$.
Thus,
${\mathfrak L}$ is linearly independent. Finally, since $\# \, {\mathfrak L} = q^s = \dim_{\F} \F[T_1,\dots,T_s]_{\le (q-1,\dots,q-1)}$,
the lemma is proved.
\end{proof}

\subsection{Non-forbidden monomials}
\label{subsec2}
Let us fix a positive integer $r\le \ell$. From Theorem~\ref{dualdescription}, we know that
$\Cr^{\perp} = \Ev\left(\fstar\right)$, where $\fstar$ is the 
space spanned by 
the non-forbidden monomials and the binomials, or more precisely, by $\basis$.
Let $\fstarmin$ denote the set of all $f\in \fstar$ such that $\Ev(f)$ is a minimum-weight codeword of $\Cr^{\perp}$, and let
$\left\langle \fstarmin \right\rangle$
denote the subspace of $\fstar$ 
spanned by $\fstarmin$.
%

We begin with a useful characterization of the non-forbidden monomials. To this end, let us first make a definition. We say that
a reduced monomial $\mu \in \RMon$ is  \emph{maximal non-forbidden} 
with respect to $\Cr$ if 
\begin{equation}
\label{eq:maxnfdlevelr}
{\rm (i)} \  \mu = \frac{\full}{t} \quad \text{ for some } t\in \Term\left(\mathcal{M}\right) \text{ and } \mathcal{M}\in \Delta_{r+1}(\ell,m), 
\end{equation}
or if there are $(i_1,j_1), (i_2,j_2) \in \rec$ 
such that
\begin{equation}
\label{eq:maxnfd}
{\rm (ii)} \
\mu = \frac{\full}{X_{i_1j_1}X_{i_2j_2}} \quad
\text{ with $i_1=i_2$ or $j_1=j_2$. } \qquad \qquad \qquad    
\end{equation}

It may be noted that in (ii) above, the possibility $(i_1,j_1) = (i_2,j_2)$ is not excluded except when $q=2$,
in which case it is automatically excluded since $\mu$ is a monomial to begin with. It may also be noted that when $r=\ell$, i.e., in the case of affine Grassmann codes, possibility (i) does not arise at all, whereas when $r=1$, we can combine (i) and (ii) to simply say that
$\mu$ is a reduced monomial of degree $\delta(q-1)-2$.
The terminology in the above definition is justified by the following. 

\begin{lemma}
\label{lem:nonforbdividesmaxnonforb}
A reduced monomial in $\RMon$ is non-forbidden 
with respect to $\Cr$ if and only if it divides some maximal non-forbidden monomial 
with respect to
$\Cr$.
\end{lemma}

\begin{proof}
%
For a monomial $\mu\in \Mon$ and for $1\le i\le \ell$ and $1\le j \le \ell'$, let us denote by $\mu_i$ the $i^{\rm th}$ row-degree of $\mu$ (i.e., the number of variables, counting multiplicities, from the $i^{\rm th}$ row of $X$ appearing in $\mu$) and by
$\mu^j$ the $j^{\rm th}$ column-degree of $\mu$.
Observe that a monomial $\mu\in \Mon$ is a term of a minor of size $\le r$, i.e., 
$\mu \in \Term({\mathcal M})$ for some ${\mathcal M}\in \Delta(\ell,m;r)$, if and only if
$\deg(\mu)\le r$, $\mu_i \le 1$ for all $i=1, \dots , \ell$ and  $\mu^j \le 1$ for all $j=1, \dots , \ell'$.
Hence if $\mu \in \RMon$ is a reduced monomial, then
$$
\mu \text{ is forbidden}  \Leftrightarrow \deg(\mu)\ge \delta(q-1)-r, \ \mu_i \ge \delta (q-1) - 1 \ \forall \, i, \text{ and } \mu^j \ge \delta(q-1)-1 \ \forall j. 
$$
In other words, a reduced monomial $\mu \in \RMon$ is non-forbidden with respect to $\Cr$ if and only if (a) $\deg(\mu) \le \delta(q-1)-(r+1)$, or
(b) $\mu_i \le \ell' (q-1)-2$ for some $i \in \{1, \dots , \ell\}$, or (c) $\mu^j \le \ell (q-1)-2$ for some $j\in \{1, \dots , \ell'\}$.
To conclude, it suffices to observe that for any $\mu \in \RMon$, we have the following.
If $\mu$ divides a monomial satisfying \eqref{eq:maxnfdlevelr}, then (a) holds. On the other hand, if  (a) holds but neither (b) nor (c) holds, then
$\mu$ divides a monomial satisfying \eqref{eq:maxnfdlevelr}. Finally, $\mu$  divides a monomial satisfying \eqref{eq:maxnfd} if and only if (b) or (c) holds.
\end{proof}

We will now proceed to show that non-forbidden monomials of type (ii), i.e., those that divide a maximal non-forbidden monomial given by \eqref{eq:maxnfd}, are generated by the minimum-weight codewords.
In what follows we will tacitly use the obvious fact that the (permutation) automorphisms of a code and its dual are identical and that
minimum-weight codewords are always preserved by an automorphism. Furthermore, we will make frequent use of the automorphisms of $\Cr$
given by Lemma~\ref{lem:sigmauAB}, 
i.e.,  the automorphisms induced by the transformation $X\mapsto BXA + {\mathbf{u}}$, where
$B\in \GL_{\ell}\left(\Fq\right)$, $A\in \GL_{\ell'}\left(\Fq\right)$ and  ${\mathbf{u}}\in M_{\ell\times\ell'}(\F)$.
It is convenient to treat the binary and the non-binary cases separately.

\begin{lemma}
\label{lem:binnfd}
Assume that $q=2$ and $\ell'>1$. Suppose $\mu \in \RMon$ is as in \eqref{eq:maxnfd} and  $\nu \in \RMon$ divides $\mu$. Then
$\nu \in \left\langle \fstarmin \right\rangle$.
\end{lemma}

\begin{proof}
First, observe that $\mu \in \fstarmin$, thanks to Theorem \ref{mindistCrdual}.
We use (finite) induction on $d:= \deg\left(\mu/\nu\right) = \deg(\mu) - \deg(\nu)$ to show that $\nu \in \left\langle \fstarmin \right\rangle$.   If $d=0$, then $\nu = \mu$ and there is nothing to prove. Assume that $d>0$ and the result holds for smaller values of $d$. Since $d>0$, there is a variable $X_{ij}$ that divides $\mu/\nu$. Write $\nu' = \nu X_{ij}$. By induction hypothesis $\nu' \in \left\langle \fstarmin \right\rangle$. Hence the polynomial, say $f'$, obtained from $\nu'$ when $X$ is changed to $X+\mathbf{u}$ is in  $\left\langle \fstarmin \right\rangle$ for every
${\mathbf{u}}\in M_{\ell\times\ell'}(\F)$. Now take ${\mathbf{u}}$ to be the $\ell\times\ell'$ matrix whose $(i,j)^{\rm th}$ entry is $1$ and all other entries are zero. 
Then $f'=\nu'+\nu$, and so $\nu 
\in  \left\langle \fstarmin \right\rangle$. 
\end{proof}

\begin{lemma}
\label{lem:nonbinfd}
Assume that $q>2$. If $f\in \rp$ is such that $\deg_{X_{ij}} f \le q-3$ for some $(i,j)\in \rec$, then $f \in  \left\langle \fstarmin \right\rangle$.
\end{lemma}

\begin{proof}
Applying an 
automorphism induced by $X\mapsto BXA$, where $B\in \GL_{\ell}\left(\Fq\right)$ and $A\in \GL_{\ell'}\left(\Fq\right)$ are suitable
permutation matrices, we may assume, without loss of generality, that $(i,j)= (\ell, \ell')$.
In view of Corollary~\ref{cor:Mqminus1T}, Remark~\ref{rem:MdT} and Lemma~\ref{lem:mdts}, we see that the space $\Fq[X]_{\le (q-1, q-1, \dots , q-1, q-3)}$ of all reduced polynomials of degree $\le q-3$ in the last variable $X_{\ell\ell'}$ is spanned by the
products of the form
\begin{equation*}
\frac{1}{\left(X_{\ell\ell'}-a_1\right) \left(X_{\ell\ell'}-a_2\right)} \prod_{(i,j)\in \rec} \left(X_{ij} - \alpha_{ij}\right) ^{q-1}-1,
\end{equation*}
where $\alpha_{ij}$ vary over $\Fq$ and $a_1, a_2$ vary over $\Fq\setminus\left\{ \alpha_{\ell\ell'}\right\}$.
But these products are precisely of the form
\eqref{eq:g} up to an automorphism induced by $X\mapsto X +\mathbf{u}$, where ${\mathbf{u}}\in M_{\ell\times\ell'}(\F)$. Hence from
Theorem \ref{mindistCrdual}, we obtain the desired result.
\end{proof}

The above lemma 
shows that if $q>2$ and if a reduced monomial $\nu$ divides a maximal non-forbidden monomial of the form
$\full/X_{ij}^2$ for some $(i,j)\in \rec$, then $\nu$ is generated by minimum-weight codewords of $\Cr^{\perp}$. This covers, in particular, the case when
$\ell'=1$ (so that $r = \ell = 1$).
It only remains to consider the case of reduced monomials dividing maximal non-forbidden monomials of the form ${\full}/{X_{i_1j_1}X_{i_2j_2}}$,
where $(i_1,j_1), (i_2,j_2)$ are distinct elements of $\rec$ and where $q>2$.

\begin{lemma}
\label{lem:nonbinfd2}
Assume that $q>2$ and $\ell'>1$.
Suppose $\mu$ is a maximal non-forbidden monomial of the form ${\full}/{X_{i_1j_1}X_{i_2j_2}}$,
where $(i_1,j_1), (i_2,j_2)$ are distinct elements of $\rec$ such that $i_1=i_2$ or $j_1=j_2$. Then
every divisor of $\mu$ is in $\left\langle \fstarmin \right\rangle$.
\end{lemma}

\begin{proof}
First, suppose $i_1=i_2$.
Applying an automorphism induced by $X\mapsto BXA$, where $B\in \GL_{\ell}\left(\Fq\right)$ and $A\in \GL_{\ell'}\left(\Fq\right)$ are suitable permutation matrices, 
we can and will assume that $\mu = \full/X_{\ell \ell'-1} X_{\ell \ell'}$. Let $\rec':= \rec \setminus \left\{ (\ell , \ell'-1), (\ell , \ell')\right\}$ and let
$$
\nu(X) = \prod_{(i,j)\in \rec'} X_{ij}^{e_{ij}} \qquad (0\le e_{ij} \le q-1)
$$
be any reduced monomial in the $\ell\ell' - 2$ variables $\left\{X_{ij} : (i,j)\in \rec'\right\}$.
By Lemma~\ref{lem:nonbinfd}, 
$\nu(X)  X_{\ell \ell'-1}^{q-3} X_{\ell \ell'}^{q-1} \in \left\langle \fstarmin \right\rangle$. Consider the
$\ell \times \ell'$ matrix $Y= \left(Y_{ij}\right)$ obtained from $X$ by adding the $(\ell'-1)^{\rm th}$ column to the last column (so that
for $1\le i \le \ell$, $Y_{ij} = X_{ij}$ if $1\le j<\ell'$ and $Y_{i\ell'} = X_{i\ell'-1}+X_{i\ell'}$). Clearly $Y$ is obtained from $X$ upon multiplication by an elementary matrix in $\GL_{\ell'}(\Fq)$ on the right, and hence $X\mapsto Y$ induces an automorphism of $\Cr$. Consequently,
the corresponding reduced polynomial is generated by the minimum-weight codewords, i.e.,
$$
\overline{\nu(Y)  X_{\ell \ell'-1}^{q-3} \left(X_{\ell \ell'-1} + X_{\ell \ell'}\right)^{q-1}} \in \left\langle \fstarmin \right\rangle \quad
\text{where} \quad \nu(Y) = \prod_{(i,j)\in \rec'} Y_{ij}^{e_{ij}}.
$$
Now since $\nu(Y)$ and $X_{\ell \ell'-1}^{q-3} \left(X_{\ell \ell'-1} + X_{\ell \ell'}\right)^{q-1}$ are polynomials in disjoint sets of variables,
in view of Remark~\ref{multred} and the binomial theorem, the polynomial
$$
\sum_{t=0}^{q-1} \overline{\nu(Y)} {\binom{q-1}{t}} \overline{X_{\ell \ell'-1}^{q-3+t} X_{\ell \ell'}^{q-1-t}}
$$
is in $\left\langle \fstarmin \right\rangle$. Moreover, by Lemma~\ref{lem:nonbinfd}, each term in the above expansion, except possibly the term corresponding to $t=1$, is in $\left\langle \fstarmin \right\rangle$. It follows therefore that the the term corresponding to $t=1$ is also in $\left\langle \fstarmin \right\rangle$. In other words, $\overline{\nu(Y)}X_{\ell \ell'-1}^{q-2} X_{\ell \ell'}^{q-2} \in \left\langle \fstarmin \right\rangle$. Finally, since the $Y_{ij}$, $(i,j)\in \rec'$, are clearly linearly independent, it follows from Lemma~\ref{lem:basis2} that polynomials of the form $\overline{\nu(Y)}$ form a basis of the space of reduced polynomials in $\left\{X_{ij} : (i,j)\in \rec'\right\}$. Hence we conclude that
any divisor of $\mu$  is in $\left\langle \fstarmin \right\rangle$. The case when $j_1=j_2$ is proved similarly.
\end{proof}

\begin{corollary}
\label{cor:nonforbifull}
Every non-forbidden monomial with respect to $\C$ is in the $\Fq$-linear span of minimum-weight codewords of $\C$.
\end{corollary}

\begin{proof}
We have noted already that when $r=\ell$, the only maximal non-forbidden monomials with respect to $\Cr$ are those of type (ii), i.e., those given by
\eqref{eq:maxnfd}. Hence the desired result follows from Lemmas \ref{lem:nonforbdividesmaxnonforb}, \ref{lem:binnfd}, \ref{lem:nonbinfd}, and \ref{lem:nonbinfd2}.
\end{proof}

\subsection{Binomials}
\label{subsec3}
Fix a positive integer $r\le \ell$. Recall that the basis $\basis$ of $\Cr^{\perp}$ consists of the non-forbidden monomials with respect to $\Cr$ and the binomials
$$
\BtM := \frac{\full}{\teM} - \frac{\full}{\tsM}, 
$$
where ${\mathcal M}$ varies over the minors of $X$ with $\deg\left({\mathcal M}\right)\le r$ and 
$\sigma$ varies over the nonidentity permutations of $\left\{1, 2, \dots , \deg\left({\mathcal M}\right)\right\}$.
The two monomials appearing in such a binomial are forbidden and therefore do not correspond to a codeword of $\Cr^{\perp}$. However,
the binomials themselves correspond to codewords of $\Cr^{\perp}$, and we will show that they are generated by the minimum-weight codewords. We begin with an elementary algebraic observation, which will be useful in the sequel.

\begin{lemma}
\label{lem:xyzw}
Let $x,y,z,w$ be independent indeterminates over $\Fq$. Consider the polynomial $f = (zw)^{q-2}(x+z)^{q-1}(y+w)^{q-1}$. Also let
\begin{equation}
\label{eq:Bnot}
\full_0:= \left(xyzw\right)^{q-1} \quad \text{and} \quad \Bnot:= \frac{\full_0}{xw} + \frac{\full_0}{yz}.
\end{equation}
Then the reduced polynomial corresponding to $f$ is given by
$$
\bar{f} = \frac{\full_0}{zw} + \frac{\full_0}{xy} - \Bnot + h,
$$
where $h\in \Fq[x,y,z,w]$
is a reduced polynomial 
such that every $\mu \in \Term(h)$ satisfies $\deg_x \mu \le q-3$ or $\deg_y \mu \le q-3$.
\end{lemma}

\begin{proof}
Expanding $(x+z)^{q-1}$ and $(y+w)^{q-1}$ by the binomial theorem, we see that
$$
\bar{f} = \sum_{s=0}^{q-1} \sum_{t=0}^{q-1} {\binom{q-1}{s}} {\binom{q-1}{t}} \overline{x^{q-1-s} z^{q-2+s}y^{q-1-t} w^{q-2+t}}.
$$
Considering separately the terms in the double summation above 
corresponding to $(s,t)=(0,0), (1,0), (0,1)$ and $(1,1)$, and upon letting $h$ denote the sum of the remaining terms, we readily
obtain the desired result.
\end{proof}

\begin{lemma}
\label{lem:singletranspo}
Assume that $r>1$. Let ${\rho}$ be an integer such that $1<{\rho}\le r$ and let ${\mathcal M}\in \Delta_{\rho}(\ell,m)$. If $\sigma, \tau\in S_{\rho}$ are such that $\sigma^{-1}\tau$ is a transposition, then 
$$
\frac{\full}{\tsM} - \frac{\full}{\ttM} \in \left\langle \fstarmin \right\rangle. 
$$
\end{lemma}

\begin{proof}
Applying an automorphism induced by
$X\mapsto BXA$, where $B\in \GL_{\ell}\left(\Fq\right)$ and $A\in \GL_{\ell'}\left(\Fq\right)$ are suitable permutation matrices, we may
assume that ${\mathcal M} = {\mathcal L}_{\rho}$, i.e., ${\mathcal M}$ is the ${\rho}^{\rm th}$ leading principal minor of $X$. Next, by a similar trick, we may assume that
$\sigma$ is the identity permutation and $\tau$ is a transposition in $S_{\rho}$, say $(st)$. Let us denote the indeterminates $X_{ss}, X_{st}, X_{ts}$ and $X_{tt}$ by
$x,y,z$ and $w$, respectively. Also let $\rec':= \rec \setminus\left\{(s,s), (s,t), (t,s), (t,t)\right\}$. With these simplifications and notations, 
\begin{equation}
\label{eq:Bnot1}
\frac{\full}{\tsM} - \frac{\full}{\ttM} = \Bnot \prod_{(i,j)\in \rec'} X_{ij}^{q-1},
\end{equation}
where $\Bnot$ is as in \eqref{eq:Bnot}. On the other hand, by Lemmas  \ref{lem:binnfd} and \ref{lem:nonbinfd2}, any divisor of $\full/zw$ is in
$\left\langle \fstarmin \right\rangle$. In particular,  $\left(zw\right)^{q-2} \left(xy\right)^{q-1} \nu(X) \in \left\langle \fstarmin \right\rangle$
for any reduced monomial $\nu (X)$ in the $\ell\ell' - 4$ variables $\left\{X_{ij} : (i,j)\in \rec'\right\}$. Now if $Z=\left(Z_{ij}\right)$ is the
$\ell \times \ell'$ matrix obtained from $X$ by adding the $t^{\rm th}$ row to the $s^{\rm th}$ row, then $X\mapsto Z$ induces an automorphism of
$\Cr$ and therefore in view of Remark~\ref{multred},
$$
\left(\overline{\left(zw\right)^{q-2} \left(x+z\right)^{q-1}\left(y+w\right)^{q-1}} \right)\overline{\nu(Z)}   \in \left\langle \fstarmin \right\rangle .
$$
Moreover, by Lemma~\ref{lem:basis2}, the reductions $\overline{\nu(Z)}$ form a basis of the space of reduced polynomials in $\left\{X_{ij} : (i,j)\in \rec'\right\}$. Consequently, $\overline{\nu(Z)}$ can be replaced by an arbitrary reduced monomial in
$\left\{X_{ij} : (i,j)\in \rec'\right\}$, and, in particular, by $\prod_{(i,j)\in \rec'} X_{ij}^{q-1}$. This, in view of Lemma~\ref{lem:xyzw}, shows that
\begin{equation}
\label{eq:Bnot2}
\frac{\full}{zw} + \frac{\full}{xy}  - \left(\Bnot \prod_{(i,j)\in \rec'} X_{ij}^{q-1}\right) + H \in \left\langle \fstarmin \right\rangle ,
\end{equation}
where $H\in \Fq[X]$ is a reduced polynomial each of whose term has $X_{ss}$-degree or $X_{st}$-degree $\le q-3$. By Lemmas  \ref{lem:binnfd} and \ref{lem:nonbinfd2}, the first two terms in the above sum  are in $\left\langle \fstarmin \right\rangle$ and moreover, so is $H$, 
thanks to Lemma~\ref{lem:nonbinfd}. It is now clear that \eqref{eq:Bnot1} and \eqref{eq:Bnot2} yield the desired result.
\end{proof}

An application of a classical result concerning permutations now yields the main result of this subsection.

\begin{lemma}
\label{lem:productoftranspo}
Let ${\rho}$ be a nonnegative integer $\le r$ and let ${\mathcal M}\in \Delta_{\rho}(\ell,m)$ and $\sigma \in S_{\rho}$. Then the binomial $\BtM$ is
in $\left\langle \fstarmin \right\rangle$.
\end{lemma}

\begin{proof}
If $\rho \le 1$, then $\sigma$ is necessarily the identity permutation $\epsilon$ and $\BtM=0$. Now assume that $\rho > 1$ and $\sigma \ne \epsilon$. Then $\sigma$ is a nonempty product of transpositions in $S_{\rho}$, say $\sigma = \tau_1\tau_2\cdots \tau_t$. Define $\sigma_0: = \epsilon$
and  $\sigma_i = \tau_1\tau_2\cdots \tau_i$ for $1\le i \le t$. Then $\sigma_{i-1}^{-1}\sigma_i$ is a transposition for $1\le i \le t$, and hence using Lemma~\ref{lem:singletranspo}, we see that
$$
\BtM = \frac{\full}{\teM} - \frac{\full}{\tsM} = \sum_{i=1}^t \; \frac{\full}{{\mathsf{t}}_{\sigma_{i-1}}({\mathcal{M}})} - \frac{\full}{{\mathsf{t}}_{\sigma_{i}}({\mathcal{M}})} 
$$
is in $\left\langle \fstarmin \right\rangle$.
\end{proof}

We are now ready to prove the main result of this section.

\begin{theorem}
\label{thm:DingKeyForAGC}
$\C^{\perp}$ is generated by its minimum-weight codewords.
\end{theorem}

\begin{proof}
Follows from Corollary~\ref{cor:nonforbifull} and Lemma~\ref{lem:productoftranspo}.
\end{proof}

In the discussion before \S \ref{subsec1}, we have noted that
${C^\mathbb{A}(\ell,m;0)}^{\perp}$ is generated by its minimum-weight codewords. Moreover, analyzing the proofs of the results in this section,
it can be seen that  ${C^\mathbb{A}(\ell,m;1)}^{\perp}$ is generated by its minimum-weight codewords. It is, however, easier to
derive the result for ${C^\mathbb{A}(\ell,m;1)}^{\perp} = \RM(1, \delta)^{\perp} = \RM(\delta(q-1)-2, \, \delta)$ directly from
Theorem~\ref{thm:DingKeyForAGC} as shown below.

\begin{corollary}
For any positive integer $d$, the Reed-Muller codes $\RM(1, d)$ and $\RM(d(q-1)-2, \, d)$ are linear codes generated by their minimum-weight codewords.
\end{corollary}

\begin{proof}
Taking $r =\ell =1$ and $\ell'=d$ in Proposition~\ref{AffGC}, we see that $\RM(1, d)$ is generated by its minimum-weight codewords. Moreover taking
$\ell =1$ and $\ell'=d$ in Theorem~\ref{thm:DingKeyForAGC}, we see that $C^\mathbb{A}(1, d+1)^{\perp} = \RM(1,d)^{\perp} = \RM(d(q-1)-2, \, d)$
is generated by its minimum-weight codewords.
\end{proof}

\begin{remark}
\label{rem:otherlevels}
For the intermediate levels, generation by minimum-weight codewords is not true, in general. More precisely, if $1< r < \ell$, then
the minimum-weight codewords of $\Cr^{\perp}$ need not generate $\Cr^{\perp}$. For example, if $\ell=\ell'=3$ and $q=r=2$, then
the affine Grassmann code $C^\mathbb{A}(3,6;2)$ is a $[512,19,192]$-code, while its dual is a $[512,493,4]$-code, and
a computer verification shows that the number of codewords of weight 4 in $C^\mathbb{A}(3,6;2)^{\perp}$ and $C^\mathbb{A}(3,6;3)^{\perp}$
is the same! Hence the minimum-weight codewords of $C^\mathbb{A}(3,6;2)^{\perp}$ just generate $C^\mathbb{A}(3,6;3)^{\perp}$. In general, we have
$$
\C^{\perp} = C^\mathbb{A}(\ell,m;\ell)^{\perp} \subset C^\mathbb{A}(\ell,m;\ell-1)^{\perp}  \subset \cdots \subset
C^\mathbb{A}(\ell,m;2)^{\perp} \subset C^\mathbb{A}(\ell,m;1)^{\perp}
$$
and it seems plausible that for $1<r<\ell$, the minimum of weight codewords of $\Cr^{\perp}$ generate the smallest of these codes, namely, $\C^{\perp}$. In fact, the results of this section seem to show that the binomials and the non-forbidden monomials of type (ii) are generated by
the minimum-weight codewords of $\Cr^{\perp}$ for \emph{any} $r=1, \dots , \ell$. In particular, they are generated by the minimum-weight codewords of $\C^{\perp}$. The difficulty arises due to maximal non-forbidden monomials of type (i), i.e., those given by \eqref{eq:maxnfdlevelr}.
At any rate, a complete determination of the minimum-weight codewords of duals of affine Grassmann codes of any level and of the space generated by them could be an interesting problem.
\end{remark}

\section*{Acknowledgments}

\label{secAck}
\begin{small}
We are grateful to the Otto M{\o}nsted Foundation, which supported the visit of Sudhir Ghorpade to the Technical University of Denmark during
May-July 2010 when some of this work was carried out.
\end{small}

\end{document}